\documentclass[11pt]{article}
\usepackage{amssymb}
\usepackage{amsmath}
\usepackage{enumerate,xspace,multirow,theorem}
\usepackage{graphicx}	

\newtheorem{theorem}{Theorem}[section]
\newtheorem{corollary}[theorem]{Corollary}
\newtheorem{lemma}[theorem]{Lemma}

\newtheorem{definition}[theorem]{Definition}
{\theorembodyfont{\rmfamily} 
}

\def\blksquare{\rule{2mm}{2mm}}
\def\qedsymbol{\blksquare}
\newcommand{\bg}[1]{\medskip\noindent{\it #1}}
\newcommand{\ed}{{\hfill\qedsymbol}\medskip}
\newenvironment{proof}{\bg{Proof} : }{\ed}
\newenvironment{proofof}[1]{\bg{Proof of {#1}} : }{\ed}

\newcommand{\np}{{\em NP}\xspace}
\newcommand{\nphard}{\np-hard\xspace}

\DeclareMathOperator{\E}{E}

\newcommand{\R}{\ensuremath{\mathbb R}}
\newcommand{\Z}{\ensuremath{\mathbb Z}}

\newcommand{\A}{\ensuremath{\mathcal{A}}}

\newcommand{\I}{\ensuremath{\mathcal I}}
\newcommand{\F}{\ensuremath{\mathcal F}}
\newcommand{\D}{\ensuremath{\mathcal D}}
\newcommand{\mE}{\ensuremath{\mathcal E}}
\newcommand{\T}{\ensuremath{\mathcal T}}

\newcommand{\Pc}{\ensuremath{\mathcal P}}

\newcommand{\OPT}{\ensuremath{\mathit{OPT}}}

\newcommand{\sm}{\ensuremath{\setminus}}
\newcommand{\es}{\ensuremath{\emptyset}}

\newcommand{\poly}{\operatorname{poly}}

\newcommand{\e}{\ensuremath{\epsilon}}

\newcommand{\gm}{\ensuremath{\gamma}}

\newcommand{\sse}{\subseteq}

\newcommand{\bufl}{{\textsf{Budget-UFL}}\xspace}
\newcommand{\ufl}{{\textsf{UFL}}\xspace}

\newcommand{\iopt}{\ensuremath{O^*}}
\newcommand{\tx}{\ensuremath{\tilde x}}
\newcommand{\ty}{\ensuremath{\tilde y}}

\newcommand{\tp}{\ensuremath{\tilde p}}
\newcommand{\tq}{\ensuremath{\tilde q}}
\newcommand{\tc}{\ensuremath{\tilde c}}
\newcommand{\tv}{\ensuremath{\tilde v}}

\newcommand{\hx}{\ensuremath{\hat x}}
\newcommand{\hq}{\ensuremath{\hat q}}

\newcommand{\hal}{\ensuremath{\hat\alpha}}
\newcommand{\hbeta}{\ensuremath{\hat\beta}}
\newcommand{\htht}{\ensuremath{\hat\theta}}
\newcommand{\tal}{\ensuremath{\tilde\alpha}}

\newcommand{\by}{\ensuremath{\bar y}}
\newcommand{\bv}{\ensuremath{\bar v}}
\newcommand{\bD}{\ensuremath{\overline {\mathcal D}}}
\newcommand{\tD}{\ensuremath{\widetilde {\mathcal D}}}

\newcommand{\ga}{\ensuremath{\gamma}}

\newcommand{\be}{\ensuremath{\beta}}
\newcommand{\ld}{\ensuremath{\lambda}}
\newcommand{\kp}{\ensuremath{\kappa}}
\newcommand{\al}{\ensuremath{\alpha}}
\newcommand{\tht}{\ensuremath{\theta}}

\newcommand{\dt}{\ensuremath{\delta}}

\newcommand{\sg}{\ensuremath{\sigma}}
\newcommand{\Gm}{\ensuremath{\Gamma}}
\newcommand{\Om}{\ensuremath{\Omega}}
\newcommand{\w}{\ensuremath{\omega}}

\newcommand{\bon}{\ensuremath{\boldsymbol{1}}}

\newcommand{\pub}{\operatorname{pub}}

\newcommand{\bc}{\ensuremath{\overline c}}

\newcommand{\paym}{PayM\xspace}
\newcommand{\cm}{CM\xspace}

\newcommand{\newlp}{\textnormal{(\ref{primal}')}\xspace}
\newcommand{\indprim}{\textnormal{(K-P)}\xspace}
\newcommand{\indnewlp}{\textnormal{(K-P')}\xspace}
\newcommand{\range}{\ensuremath{R}}
\newcommand{\lp}{\ensuremath{\mathrm{LP}}}
\newcommand{\ext}{\ensuremath{\mathrm{ext}}}
\newcommand{\ex}{\ext}
\newcommand{\flp}{\textnormal{(FL-P)}\xspace}
\newcommand{\flopt}{\ensuremath{\OPT_{\text{FL-P}}}\xspace}
\newcommand{\proj}{\ensuremath{\mathsf{proj}}}

\title{Near-Optimal and Robust Mechanism Design for Covering Problems with Correlated 
Players\footnote{A preliminary version appeared as~\cite{MinooeiS13}.}}   

\author{
         Hadi Minooei\thanks{{\tt\{hminooei,cswamy\}@uwaterloo.ca}.
         Dept. of Combinatorics and Optimization, Univ. Waterloo, Waterloo, ON N2L 3G1.
         Research supported partially by NSERC grant 327620-09 and the second author's
         Discovery Accelerator Supplement Award, and Ontario Early Researcher Award.}
\and
         Chaitanya Swamy $^{\text{\thefootnote}}$
}
           
\date{}

\begin{document}

\maketitle

\vspace*{-3ex}

\begin{abstract}
We consider the problem of designing incentive-compatible, ex-post individually 
rational (IR) mechanisms for covering problems in the Bayesian setting, where players'
types are drawn from an underlying distribution and may be correlated, and the goal is to 
minimize the expected total payment made by the mechanism. 
We formulate a notion of incentive compatibility (IC) that we call 
{\em support-based IC} that is substantially more robust than Bayesian IC, 
and develop black-box reductions from support-based-IC mechanism design to algorithm design.
For single-dimensional settings, this black-box reduction applies even when we only have an
LP-relative {\em approximation algorithm} for the algorithmic problem. 
Thus, we obtain near-optimal mechanisms for various covering settings including 
single-dimensional covering problems, multi-item procurement auctions, and
multidimensional facility location.
\end{abstract}

\section{Introduction}
In a {\em covering mechanism-design problem}, there are players who provide covering
objects and a buyer who needs to obtain a suitable collection of objects so as to
satisfy certain covering constraints (e.g., covering a ground set); 
each player incurs a certain private cost, which we refer to as his {\em type}, for
providing his objects, and the mechanism must therefore pay the players from whom objects 
are procured. We consider the problem of designing incentive-compatible, ex-post
individually rational (IR) mechanisms for {covering problems} (also called 
{\em procurement auctions}) in the {\em Bayesian setting}, where players' types are drawn
from an underlying distribution and may be {\em correlated}, and the goal is to 
{\em minimize the expected total payment made by the mechanism}. 
Consider the simplest such setting of a {\em single-item procurement auction}, where a
buyer wants to buy an item from any one of $n$ sellers. Each seller's private type is the 
cost he incurs for supplying the item and the sellers must
therefore be incentivized via a suitable payment scheme.
Myerson's seminal result~\cite{M81} solves this problem (and other single-dimensional
problems) when players' private types are {\em independent}. However, 
no such result (or characterization) is known when players' types are correlated.
This is the question that motivates our work.

Whereas the analogous revenue-maximization problem for packing domains, such as
combinatorial auctions (CAs), has been extensively studied in the algorithmic mechanism
design (AMD) literature, both in the case of independent and correlated (even
interdependent) player-types (see,
e.g.,~\cite{CaiDW13a,CaiDW13b,CaiDW12,AlaeiFHHM12,ChawlaHMS10,HartN12,DobzinskiFK11,PapadimitriouP11,CaiDW12,RoughgardenT13}   
and the references therein), 
surprisingly, there are {\em almost no results} on the payment-minimization 
problem for covering settings in the AMD literature (see however the discussion in
``Related work'' for some exceptions). 
The economics literature does contain various general
results that apply to both covering and packing problems. However much of this work
focuses on characterizing special cases; see, e.g.,~\cite{Krishna,Klemperer}. An exception
is the work of Cr\'emer and McLean~\cite{CremerM85,CremerM88}, which shows that under
certain conditions, one can devise a Bayesian-incentive-compatible (BIC) mechanism whose
expected total payment exactly equal to the expected cost incurred by the players, albeit
one where players may incur negative utility under certain type-profile realizations.

\vspace{-1ex}
\paragraph{Our contributions.}
We initiate a study of payment-minimization (\paym) problems from the AMD perspective of
designing computationally efficient, near-optimal mechanisms. 
We develop black-box reductions from mechanism design to algorithm design whose
application yields a variety of optimal and near-optimal mechanisms.  
As we elaborate below, covering problems turn out to behave quite differently in certain
respects from packing problems, which necessitates new approaches (and solution
concepts). 

Formally, we consider the setting of {\em correlated players} in the explicit model, that
is, where we have an explicitly-specified {\em arbitrary discrete} joint distribution of
players' types. A commonly-used solution concept in Bayesian settings is Bayesian
incentive compatibility (BIC) and interim individual rationality (interim IR), wherein 
at the {\em interim} stage when a player knows his type but is oblivious of the random
choice of other players' types, truthful participation in the mechanism by all players
forms a Bayes-Nash equilibrium. 
Two serious drawbacks of this solution concept (which are exploited strikingly and
elegantly in~\cite{CremerM85,CremerM88}) are that: (i) a player may regret his decision of 
participating and/or truthtelling {\em ex post}, that is, after observing the realization
of other players' types; and (ii) it is overly-reliant on having precise knowledge of
the true underlying distribution making this a rather {\em non-robust} concept: if the true
distribution differs, possibly even slightly, from the mechanism designer's and/or players'
beliefs or information about it, then the mechanism could lose its IC and IR properties.

On the other hand, the solution concept of dominant-strategy incentive compatibility
(DSIC) and ex-post IR ensures  
that truthful participation is the best choice, and a no-regret choice, for every player
{\em regardless} of the other players' reported types. This is the most robust 
notion of IC and IR, since it is a completely distribution-independent.

We consider the problem of designing near-optimal mechanisms for payment-minimization
problems with robustness being an important consideration. This makes (DSIC, ex-post IR),
which we abbreviate to (DSIC, IR), as the natural ideal.
However, certain difficulties arise in achieving this goal for covering problems (see 
``Differences \ldots packing problems'' below).
We formulate a notion of incentive compatibility that we call 
{\em support-based IC}%
\footnote{The conference version~\cite{MinooeiS13} of this paper referred to this as
  ``robust Bayesian IC.''}  
that, while somewhat weaker than DSIC, is still substantially more robust than BIC  
and at the same time is flexible enough that it allows one to obtain various polytime
near-optimal mechanisms satisfying this notion. 
A {\em support-based-(IC, IR)} mechanism (see Section~\ref{prelim}) 
ensures that truthful participation in the mechanism is in the best interest of every
player (i.e. a ``no-regret'' choice) 
for every type profile in a certain {\em subset} of the type space: for every
player $i$, we impose the IC and IR conditions for all type profiles of the form 
$(\cdot,c_{-i})$ for every profile $c_{-i}$
of other players' types coming from the support of the underlying distribution.
In particular, a support-based-(IC, IR) mechanism ensures that truthful participation is
the best choice for every player {\em even at the ex-post stage} when the other players' 
(randomly-chosen) types are revealed to him.     
Such a mechanism is significantly more robust than a (BIC, interim-IR) 
mechanism since it retains its IC and IR properties for a large class of distributions 
that contains (in particular) every distribution whose support is a subset of the support
of the actual distribution. In other words, in keeping with Wilson's doctrine of
detail-free mechanisms, 
the mechanism functions robustly even 
under fairly limited information about the type-distribution. 

We show that for a variety of settings, one can reduce the support-based-(IC, IR)
payment-minimization (\paym) mechanism-design problem to the algorithmic
cost-minimization (\cm) problem of finding an outcome that minimizes the total
cost incurred. 
Moreover, this black-box reduction applies to: (a) single-dimensional settings even when
we only have an LP-relative {\em approximation algorithm} for the \cm problem (that 
is required to work only with nonnegative costs) (Theorem~\ref{bicredn2}); and 
(b) {\em multidimensional problems} with additive types (Corollary~\ref{bicredn1}). 

Our reduction yields near-optimal support-based (IC-in-expectation, IR) mechanisms for a 
variety of covering settings such as
(a) various single-dimensional covering problems including single-item procurement
auctions (Table~\ref{onedapps}); 
(b) multi-item procurement auctions (Theorem~\ref{multiproc}); and 
(c) multidimensional facility location (Theorem~\ref{uflres}).
(Support-based IC-in-expectation means that the support-based-IC guarantee holds for the 
{\em expected utility} of a player, where the expectation is over the random coin tosses 
{of the mechanism}.) 
In Section~\ref{extn}, we consider some extensions involving both weaker 
(but more robust than (BIC, interim IR)) solution concepts and the stronger (DSIC, IR)
solution concept. We obtain the same guarantees under the various weaker 
solution concepts, and adapt our techniques to obtain (DSIC-in-expectation, IR) mechanisms
with the same guarantees for single-dimensional problems in time exponential in the number   
of players (Section~\ref{dsic}).
These are the first results for the \paym mechanism-design problem for covering
settings with correlated players under a notion stronger than (BIC, interim IR). To our
knowledge, our results are new even for the simplest covering setting of single-item
procurement auctions.

\vspace{-1ex}
\paragraph{Our techniques.}
The starting point for our construction is a linear-programming (LP) relaxation
\eqref{primal} for the problem of computing an optimal support-based-(IC,
IR)-in-expectation mechanism.  
This was also the starting point in the work of~\cite{DobzinskiFK11}, which considers the  
revenue-maximization problem for CAs, but the covering nature 
of the problem 
makes it difficult to apply 
certain techniques utilized successfully in the context of packing problems (as described
below). 

We show that an optimal solution to \eqref{primal} can be computed given an optimal
algorithm $\A$ for the \cm problem since $\A$ can be used to obtain a
separation oracle for the dual LP. Next, we prove that a feasible solution
to \eqref{primal} can be extended to a support-based-(IC-in-expectation, IR) mechanism
with no larger objective value.

For single-dimensional problems, we show that even LP-relative $\rho$-approximation
algorithms for the \cm problem can be utilized, as follows.
We move to a relaxation of \eqref{primal}, where we replace the set of allocations with
the feasible region of the \cm-LP. This can be solved efficiently, since the separation
oracle for the dual can be obtained by optimizing over the feasible region of \cm-LP,
which can be done efficiently! But now we need to work harder to ``round'' 
an optimal solution $(x,p)$ to the relaxation of \eqref{primal} and obtain a
support-based-(IC-in-expectation, IR) mechanism. Here, we exploit the Lavi-Swamy~\cite{LaviS11} 
convex-decomposition procedure, using which we can show (roughly speaking) that we can
decompose $\rho x$ into a convex combination of allocations. This allows us to
obtain a support-based-(IC-in-expectation, IR) mechanism while blowing up the payment by a 
$\rho$-factor. 

In comparison with the reduction in~\cite{DobzinskiFK11}, which is the work most
closely-related to ours, our reduction from support-based-IC mechanism design to the
algorithmic \cm problem is stronger in the following sense.
For single-dimensional settings, it applies even with LP-relative {\em approximation
algorithms}, and the approximation algorithm is required to work only for ``proper inputs''
with nonnegative costs. (Note that whereas for packing problems, allowing negative-value
inputs can be benign, this can change the character of a covering problem considerably.)
In contrast, Dobzinski et al.~\cite{DobzinskiFK11} require an exact algorithm for the
analogous social-welfare-maximization (SWM) problem.   

\vspace{-1ex}
\paragraph{Differences with respect to packing problems.}
At a high level, our method of writing an LP for the underlying mechanism-design
problem and solving it given an algorithm for an associated algorithmic problem is similar  
to the procedure in Dobzinski et al.~\cite{DobzinskiFK11}. However, we encounter three
distinct sources of difficulty when dealing with covering problems vis-a-vis packing
problems.  

First, as noted in~\cite{DobzinskiFK11}, the LP can only encode the IC and IR conditions 
for a finite set of type profiles, whereas, with an infinite type space, both
support-based-(IC, IR) and (DSIC, IR) require the mechanism to satisfy 
the IC and IR conditions for an infinite set of type profiles. 
Therefore, to translate the LP solution to a suitable mechanism, we need to 
solve the ``extension problem'' of extending an allocation and pricing rule
defined on a (finite) subset of the type space to the entire type space while preserving
its IC and IR properties. This turns out to be a much more difficult task for covering
problems than for packing domains. 
The key difference is that for a packing setting such as combinatorial auctions, one can 
show that {\em any} LP solution%
---in particular, the optimal LP solution---can be converted into a (DSIC-in-expectation,
IR) mechanism without any loss in expected revenue (see Section~\ref{append-packextn}).
(Consequently, \cite{DobzinskiFK11} obtain (DSIC-in-expectation, IR) mechanisms.)
Intuitively, this works because one can focus on a single player by
allocating no items to the other players. 
Clearly, one cannot mimic this for covering problems: dropping players may render the
problem infeasible, and it is not clear how to extend an LP-solution to a
(DSIC-in-expectation, IR) mechanism for covering problems.
We suspect that not every LP solution or support-based-(IC, IR) mechanism can be
extended to a (DSIC, IR) mechanism, and 
that there is a gap between the optimal expected total payments of
support-based-(IC-in-expectation, IR) and (DSIC, IR) mechanisms. We leave these as open 
problems.%
\footnote{We show in Section~\ref{dsic} that if we expand our LP to include IC and IR
constraints for a much larger (exponential-sized) set of type-profiles then, for
single-dimensional settings, it is possible to extend an LP-solution to a
(DSIC-in-expectation, IR) mechanism. 
} 

Due to this complication, we sacrifice a modicum of the IC, IR properties in favor of
obtaining polytime near-optimal mechanisms and settle for the weaker, but still quite
robust notion of support-based (IC-in-expectation, IR). We consider this to be a reasonable 
starting point for exploring mechanism-design solutions for covering problems, which 
leads to various interesting research directions.%
\footnote{For covering problems, even formulating the LP in a way that its solution can be
extended to a support-based-(IC, IR) mechanism is somewhat tricky due to the 
following complication. 
We need to argue (see Lemma~\ref{estim}) that there is an optimal
support-based-(IC-in-expectation, IR) mechanism such that for every player $i$, and every
profile $c_{-i}$ coming from the support of the underlying distribution,  
there is a type profile $(m_i,c_{-i})$ under which the mechanism never procures objects
from $i$, and include this condition in the LP (otherwise, the IC and IR conditions  
would force the LP-extension to make arbitrarily large payments).}  

A second difficulty, which we have alluded to above, arises due to the fact that solving
the LP requires one to solve the \cm problem with negative-valued inputs. 
This is also true of packing problems~\cite{DobzinskiFK11} (where one needs to solve the
SWM problem), even in the single-item setting~\cite{M81} 
(where reserve prices arise due to negative virtual valuations).  
While this is not a problem if we have an optimal algorithm for the \cm problem, it
creates serious issues, even in the single-dimensional setting, if we only have an
approximation algorithm at hand; in particular, the standard notion of 
approximation becomes meaningless since the optimum could be negative.
In contrast, for packing problems with single-dimensional types (or additive
types~\cite{CaiDW12,CaiDW13a,CaiDW13b}),   
these issues are more benign 
since one may always discard players (or options of players) with negative value. 
In particular, 
an approximation algorithm can be used to obtain an
approximate separation oracle for the dual LP, and thus obtain a near-optimal solution to  
the primal LP via a well-known technique in approximation algorithms. (We sketch this
extension of a result of~\cite{DobzinskiFK11} in Appendix~\ref{append-packextn}.)

Finally, a stunning aspect where covering and packing problems diverge can be seen when
one considers the idea of a $k$-lookahead auction~\cite{Ronen,DobzinskiFK11}. This was
used by~\cite{DobzinskiFK11} to convert their results in the explicit model to the 
oracle model introduced by~\cite{Ronen}. 
This however fails spectacularly in the covering setting. One can show that even for
single-item procurement auctions, dropping even a single player can lead to an arbitrarily 
large payment compared to the optimum (see Appendix~\ref{append-lookahead}).

\vspace{-1ex}
\paragraph{Other related work.}
In the economics literature, the classical results of Cr\'emer and
McLean~\cite{CremerM85,CremerM88} and McAfee and Reny~\cite{McafeeR}, also apply to
covering problems, and show that one can devise a (BIC, interim IR) mechanism with
correlated players whose expected total payment is at most the expected total cost
incurred provided the underlying type-distribution satisfies a certain full-rank
assumption. These mechanisms may however cause a player to have negative utility under
certain realizations of the random type profile.

The AMD literature has concentrated mostly on the independent-players
setting~\cite{CaiDW13a,CaiDW13b,CaiDW12,AlaeiFHHM12,ChawlaHMS10,HartN12}. There has 
been some, mostly recent, work that also considers correlated
players~\cite{Ronen,RonenL05,DobzinskiFK11,PapadimitriouP11,CaiDW12,RoughgardenT13}. 
Much of this work pertains to the revenue-maximization setting; an exception
is~\cite{RonenL05}, which is discussed below.
Ronen~\cite{Ronen} considers the single-item auction setting in the {\em oracle model},
where one samples from the distribution conditioned on some players' values. He proposes
the (1-) lookahead auction and shows that it achieves a $\frac{1}{2}$-approximation. 
Papadimitriou and Pierrakos~\cite{PapadimitriouP11} show that the optimal (DSIC, IR)
mechanism for the single-item auction can be computed efficiently with at most 2 players,
and is \nphard otherwise. 
Cai et al.~\cite{CaiDW12} give a characterization of the optimal
auction under certain settings. 
Roughgarden and Talgam-Cohen~\cite{RoughgardenT13} consider interdependent types, which
generalizes the correlated type-distribution setting, and develop an analog of Myerson's
theory for certain such settings.   

Ronen and Lehmann~\cite{RonenL05} consider the \paym problem in the setting where a buyer
wants to buy an item from sellers who can supply the item in on of many configurations and
incur private costs for supplying the item. However, this procurement problem is in fact a
{\em packing problem}: one can view a solution to be feasible if it selects at most one
configuration for procurement; in particular, the buyer has the flexibility of 
{\em not procuring} the item. As noted earlier, this flexibility drastically alters the
character of the mechanism-design problem. Not surprisingly, the results therein, which
are based on lookahead auctions, do not apply in the covering setting (as noted above). 

Various reductions from revenue-maximization to SWM are given
in~\cite{CaiDW12,CaiDW13a,CaiDW13b}. The reductions in~\cite{CaiDW12,CaiDW13b} also apply
to covering problems and the \paym objective, but they are incomparable to our results. 
These works focus on the (BIC, interim-IR) solution concept, which 
is a rather weak/liberal notion for correlated distributions.  
Most (but not all) of these consider independent players and additive
valuations, and often require that the SWM-algorithm also work with negative values, which
is a benign requirement for downward-closed environments such as CAs but is quite
problematic for covering problems when only has an approximation algorithm. 
Cai et al.~\cite{CaiDW12} consider correlated players and 
obtain mechanisms having running time polynomial in the maximum support-size of
the marginal distribution of a player, which could be substantially smaller than the
support-size of the entire distribution. 
This savings can be traced to the use of the (BIC, interim-IR) notion which
allows~\cite{CaiDW12} to work with a compact description of the mechanism.
It is unclear if these ideas are applicable when one considers robust-(BIC, IR)
mechanisms. A very interesting open question is whether one can design 
robust-(BIC-in-expectation, IR) mechanisms having running time polynomial in the
support-sizes of the marginal player distributions (as in~\cite{CaiDW12,DobzinskiFK11}).

\section{Preliminaries} \label{prelim}

\paragraph{Covering mechanism-design problems.}
We adopt the formulation in~\cite{MinooeiS12} to describe general covering
mechanism-design problems. 
There are some items that need to be covered, and $n$ players who provide covering
objects. Let $[k]$ denote the set $\{1,\ldots,k\}$. 
Each player $i$ provides a set $\T_i$ of covering objects.
All this information is public knowledge.  
Player $i$ has a {\em private cost} or {\em type} vector $c_i=\{c_{i,v}\}_{v\in\T_i}$,
where $c_{i,v}\geq 0$ is the cost he incurs for providing object $v\in\T_i$; for
$T\sse\T_i$, we use $c_i(T)$ to denote $\sum_{v\in T}c_{i,v}$. 
A feasible solution or allocation selects a subset $T_i\sse\T_i$ for each agent $i$,  
denoting that $i$ provides the objects in $T_i$, such that $\bigcup_iT_i$ covers all the
items. Given this solution, each agent $i$ incurs the private cost $c_i(T_i)$, and the
mechanism designer incurs a publicly known cost $\pub(T_1,\ldots,T_n)\geq 0$, which may be
used to encode any feasibility constraints in the covering problem. 

Let $C_i$ denote the set of all possible types of agent $i$, and 
$C=\prod_{i=1}^n C_i$. We assume (for notational simplicity) that $C_i=\R^{|\T_i|}_+$.
Let $\Om:=\{(T_1,\ldots,T_n): \pub(T_1,\ldots,T_n)<\infty\}$ be the (finite) set of all
feasible allocations.  
For a tuple $x=(x_1,\ldots,x_n)$, we use $x_{-i}$ to denote
$(x_1,\ldots,x_{i-1},x_{i+1},\ldots,x_n)$. Similarly, let $C_{-i}=\prod_{j\neq i}C_j$.  
For an allocation $\w=(T_1,\ldots,T_n)$, we sometimes use $\w_i$ to
denote $T_i$, $c_i(\w)$ to denote $c_i(\w_i)=c_i(T_i)$, and $\pub(\w)$ to denote
$\pub(T_1,\ldots,T_n)$.
We make the mild assumption that $\pub(\w')\leq\pub(\w)$ if $\w_i\sse\w'_i$ for all $i$;
so in particular, if $\w$ is feasible, then adding covering objects to the $\w_i$s
preserves feasibility. 

A (direct revelation) {\em mechanism} $M=({\cal A},p_1,\ldots,p_n)$ for a covering problem
consists of an allocation algorithm ${\cal A}: C\mapsto\Om$ and a 
payment function $p_i:C\mapsto\R$ for each agent $i$. 
Each agent $i$ reports a cost function $c_i$ (that might be different from his true cost
function). The mechanism computes the allocation $\A(c)=(T_1,\ldots,T_n)=\w \in \Om$, 
and pays $p_i(c)$ to each agent $i$. The {\it utility} $u_i(c_i,c_{-i};\bc_i)$ that 
player $i$ derives when he reports $c_i$ and the others report 
$c_{-i}$ is $p_i(c)-\bc_i(\w_i)$ where $\bc_i$ is his true cost function, and each agent $i$
aims to maximize his own utility. 
We refer to $\max_i|\T_i|$ as the dimension of a covering problem.
Thus, for a {\em single-dimensional} problem, 
each player $i$'s cost can  be specified as $c_i(\w)=c_i\al_{i,\w}$, where $c_i\in\R_+$ is
his private type and $\al_{i,\w}=1$ if $\w_i\neq\es$ and 0 otherwise.   

The above setup yields a {\em multidimensional covering mechanism-design problem} with
{\em additive types}, where by additivity we mean that 
the private cost that a player $i$ incurs 
for providing a set $T\sse\T_i$ of objects is additive across the objects in $T$ (i.e., it
is $\sum_{v\in T_i}c_{i,v}$). Notice that if $c_i,c'_i\in C_i$, then
the type $c_i+c'_i$ defined by $\{c_{i,v}+c'_{i,v}\}_{v\in\T_i}$ is also in $C_i$ and
satisfies $(c_i+c'_i)(\w)=c_i(\w)+c'_i(\w)$ for all $\w\in\Om$.
It is possible to define more general multidimensional settings, but
additive types is a reasonable starting point to explore the multidimensional
covering mechanism-design setting. (As noted earlier, there has been almost no work on
designing polytime, near-optimal mechanisms for covering problems.)   

\vspace{-1ex}
\paragraph{The Bayesian setting.}
We consider {\em Bayesian} settings where there is an underlying 
publicly-known {\em discrete} and possibly {\em correlated} joint type-distribution on $C$
from which the players' types are drawn.  
We consider the so-called explicit model, where the players' type distribution is
explicitly specified. We use $\D\subseteq C$ to denote the support of the type
distribution, and $\Pr_{\D}(c)$ to denote the probability of realization of $c\in C$.  
Also, we define $\D_i:=\{c_i\in C_i: \exists c_{-i}\text{ s.t. }(c_i,c_{-i})\in\D\}$, and 
$\D_{-i}$ to be  $\{c_{-i}: \exists c_i\text{ s.t. }(c_i,c_{-i})\in\D\}$. 

\vspace{-1ex}
\paragraph{Solution concepts.}
A mechanism sets up a game between players, and the solution concept dictates certain
desirable properties that this game should satisfy, so that one can reason about the
outcome when rational players are presented with a mechanism satisfying the solution
concept.  
The two chief properties that one seeks to capture 
relate to incentive compatibility (IC), which (roughly speaking) means that every agent's
best interest is to reveal his type truthfully, and individual rationality (IR), which is
the notion that no agent is harmed by participating in the mechanism. Differences and 
subtleties arise in Bayesian settings depending on the stage at which we impose these
properties and how robust we would like these properties to be with respect to the
underlying type distribution.

\begin{definition} \label{bic}
A mechanism $M=\bigl(\A,\{p_i\}\bigr)$ is {\em Bayesian incentive compatible} (BIC) and 
{\em interim IR} if for every player $i$ and every $\bc_i\in\D_i,c_i\in C_i$, we have 
$\E_{c_{-i}}[u_i(\bc_i,c_{-i};\bc_i)|\bc_i]
\geq\E_{c_{-i}}[u_i(c_i,c_{-i};\bc_i)|\bc_i]$ (BIC)
and
$\E_{c_{-i}}[u_i(\bc_i,c_{-i};\bc_i)|\bc_i]\geq 0$ (interim IR), where $\E_{c_{-i}}[.|\bc_i]$ denotes
the expectation over the other players' types conditioned on $i$'s type being $\bc_i$.
\end{definition}

As mentioned earlier, the (BIC, interim-IR) solution concept may yet lead to
ex-post ``regret'', and is quite non-robust in the sense that the mechanism's IC and
IR properties rely on having detailed knowledge of the distribution; 
thus, in order to be confident that a BIC mechanism achieves 
its intended functionality, one must be confident about the ``correctness'' of the
underlying distribution, and learning this information might entail significant cost. 
To remedy these weaknesses, we propose and investigate the following stronger IC and IR
notions.  

\begin{definition} \label{expostbic}
A mechanism $M=\bigl(\A,\{p_i\}\bigr)$ is {\em support-based IC} and 
{\em support-based IR}, 
if for every player $i$, every $\bc_i,c_i\in C_i$, and every
$c_{-i}\in\D_{-i}$, we have  
$u_i(\bc_i,c_{-i};\bc_i) \geq u_i(c_i,c_{-i};\bc_i)$ (support-based IC) and
$u_i(\bc_i,c_{-i};\bc_i)\geq 0$ (support-based IR). 
\end{definition}

Support-based (IC, IR) ensures that participating truthfully in the mechanism is
in the best interest of every player even at the ex-post stage 
when he knows the realized types of all players. 
To ensure that support-based IC and support-based IR are compatible, we focus on 
{\em monopoly-free} settings: for every player $i$, there is some $\w\in\Om$ with
$\w_i=\es$.  

Notice that support-based (IC, IR) is subtly weaker than the notion of (dominant-strategy IC
(DSIC), IR), wherein the IC and IR conditions of Definition~\ref{expostbic} must hold for
all $c_{-i}\in C_{-i}$, 
ensuring that truthtelling and participation are no-regret choices for a player even if
the other players' reports are outside the support of the underlying type-distribution. 
We focus on support-based IC because it forms a suitable middle-ground between BIC
and DSIC: it inherits the desirable robustness properties of DSIC, making it much
more robust than BIC (and closer to a worst-case notion), and yet is flexible
enough that one can devise polytime mechanisms satisfying this solution concept.

It might seem strange that in the definition of support-based (IC, IR), we consider
player $i$'s incentives for types outside of $i$'s support and for type-profiles that are 
inconsistent with the underlying distribution. Keeping robustness in mind,
our goal here is to approach the ideal of (DSIC, IR) and we have therefore formulated the 
most-robust notion that permits us to devise polytime near-optimal mechanisms satisfying
this notion. 
In Section~\ref{extn}, we consider various alternate solution concepts that, while weaker
than support-based (IC, IR), still retain its robustness properties to a large extent, and 
show that our results extend easily to these solution concepts.

The above definitions are stated for a deterministic mechanism, but they have
analogous extensions to a randomized mechanism $M$; 
the only change is that the $u_i(.)$ and $p_i(.)$ terms are now replaced by the expected
utility $\E_M[u_i(.)]$ and expected price $\E_M[p_i(.)]$ respectively, where the
expectation is over the random coin tosses of $M$. 
We denote the analogous solution concept for a randomized mechanism by appending 
``in expectation'' to the solution concept, e.g., a (BIC, interim IR)-in-expectation 
mechanism denotes a randomized mechanism whose expected utility 
satisfies the BIC and interim-IR requirements stated in Definition~\ref{bic}. 

A support-based-(IC, IR)-in-expectation mechanism $M=\bigl(\A,\{p_i\}\bigr)$ can be easily
modified so that the IR condition holds  
{\em with probability 1} (with respect to $M$'s coin tosses) while the expected payment to
a player (again over $M$'s coin tosses) is unchanged: on input $c$, if
$\A(c)=\w\in\Om$ with probability $q$, the new mechanism returns, with probability $q$,
the allocation $\w$, and payment $c_i(\w)\cdot\frac{\E_M[p_i(c)]}{\E_M[c_i(\w)]}$ to each
player $i$ (where we take $0/0$ to be $0$, so if $c_i(\w)=0$, the payment to $i$ is 0). 
Thus, we obtain a mechanism whose expected utility satisfies the support-based-IC condition,
and IR holds with probability 1 for all $\bc_i\in C_i,c_{-i}\in\D_{-i}$, 
A similar transformation can be applied to a (DSIC, IR)-in-expectation mechanism. 

\vspace{-1ex}
\paragraph{Optimization problems.}
Our main consideration is to minimize the expected total payment of the
mechanism. 
It is natural to also incorporate the mechanism-designer's cost into the objective. 
Define the {\em disutility} of a mechanism $M=\bigl(\A,\{p_i\}\bigr)$ under input $c$ to be
$\sum_i p_i(c)+\kp\cdot\pub\bigl(\A(c)\bigr)$, where $\kp\geq 0$ is a given scaling factor.  
Our objective is to devise a polynomial-time support-based-(IC (in-expectation), IR)-mechanism
with minimum expected disutility.
Since most problems we consider have $\pub(\w)=0$ for all feasible allocations, in
which case disutility equals the total payment, abusing terminology slightly, we refer
to the above mechanism-design problem as the {\em payment-minimization} (\paym) problem. 
(An exception is metric {\em uncapacitated facility location} (\ufl), where players
provide facilities and the underlying metric is public knowledge; 
here, $\pub(\w)$ is the total client-assignment cost of the solution $\w$.)
We always use $\iopt$ to denote the expected disutility of an optimal mechanism for the
\paym problem under consideration. 

We define the {\em cost minimization} (\cm) problem to be the {\em algorithmic}
problem of finding $\w\in\Om$ 
that minimizes the total cost $\sum_i c_i(\w)+\pub(\w)$ incurred.

The following technical lemma, whose proof we defer to Appendix~\ref{append-prelim}, 
will prove quite useful since it allows us to 
restrict the domain to a bounded set, which is essential to achieve IR with finite
prices. (For example, in the single-dimensional setting, the payment is equal to the
integral of a certain quantity from $0$ to $\infty$, and a bounded domain ensures that
this is well defined.)   
Note that such complications do not arise for packing problems.  
Let $\bon_{\T_i}$ be the $|\T_i|$-dimensional all 1s vector.
Let $\I$ denote the input size.

\begin{lemma} \label{estim}
We can efficiently compute an estimate $m_i>\max_{c_i\in\D_i,v\in\T_i}c_{i,v}$ with $\log
m_i=\poly(\I)$ for all $i$ such that there is an optimal
support-based-(IC-in-expectation, IR) mechanism 
$M^*=\bigl(\A^*,\{p^*_i\}\bigr)$ 
where $\A^*(m_i\bon_{\T_i},c_{-i})_i=\es$ with probability $1$
(over the random choices of $M^*$) 
for all $i$ and all $c_{-i}\in\D_{-i}$.
\end{lemma}

It is easy to obtain the stated estimates if we consider only {\em deterministic}
mechanisms, but it turns out to be tricky to obtain this when one allows randomized
mechanisms due to the artifact that a randomized mechanism may choose arbitrarily
high-cost solutions as long as they are chosen with small enough probability. 
In the sequel, we set $\bD_i:=\D_i\cup\{m_i\bon_{\T_i}\}$ for all $i\in[n]$, 
and $\bD:=\bigcup_i(\bD_i\times\D_{-i})$. Note that $|\bD|=O(n|\D|^2)$.

\section{LP-relaxations for the payment-minimization problem} \label{lp}
The starting point for our results is the LP \eqref{primal} that essentially encodes the
payment-minimization problem. 
Throughout, we use $i$ to index players, $c$ to index type-profiles in
$\bD$, and $\w$ to index $\Om$. 
We use variables $x_{c,\w}$ to denote the probability of choosing $\w$, and
$p_{i,c}$ to denote the expected payment to player $i$, for input $c$.
For $c\in\bD$, let $\Om(c)=\Om$ if $c\in\bigcup_i(\D_i\times\D_{-i})$, and otherwise if 
$c=(m_i\bon_{\T_i},c_{-i})$, let $\Om(c)=\{\w\in\Om: \w_i=\es\}$ (which is non-empty since
we are in a monopoly-free setting). 

\begin{alignat}{3}
\min & \quad & \sum_{c\in\D}{\textstyle \Pr_\D}(c)
\Bigl(\sum_ip_{i,c}&+\kp\sum_{\w}x_{c,\w}\pub(\w)\Bigr) \tag{P} \label{primal} \\
\text{s.t.} && \sum_{\w}x_{c,\w} & = 1 \qquad && \forall c \in {\bD} \label{e2} \\
&& p_{i,(c_i,c_{-i})} - \sum_{\w}c_i(\w)x_{(c_i,c_{-i}),\w} & \geq 
p_{i,(c'_i,c_{-i})}-\sum_{\w}c_i(\w)x_{(c'_i,c_{-i}),\w} \quad 
&& \forall i, c_i,c'_i\in\bD_i, c_{-i}\in\D_{-i} \label{e3} \\
&& p_{i,(c_i,c_{-i})} - \sum_{\w}c_i(\w)x_{(c_i,c_{-i}),\w} & \geq 0 \qquad 
&& \forall i, c_i\in\bD_i, c_{-i}\in\D_{-i} \label{e4} \\
&& p,x  \geq 0, \quad x_{c,\w} & =0 \quad && \forall c, \w\notin\Om(c). \label{e45} 
\end{alignat}
\eqref{e2} encodes that an allocation is chosen for every $c\in\bD$, and
\eqref{e3} and \eqref{e4} encode the support-based-IC and support-based-IR conditions
respectively. Lemma~\ref{estim} ensures that \eqref{primal} correctly encodes \paym,
so that $\OPT:=\OPT_{\text{\ref{primal}}}$ is a lower bound on the expected disutility
of an optimal mechanism.

Our results are obtained by computing an optimal solution to \eqref{primal}, or a
further relaxation of it, and translating this to a near-optimal 
support-based-(IC-in-expectation, IR) mechanism. Both steps come with their own challenges. 
Except in very simple settings (such as single-item procurement auctions), $|\Om|$ is
typically exponential in the input size, and therefore it is not clear how to
solve \eqref{primal} efficiently. We therefore consider the dual LP \eqref{dual}, which 
has variables $\ga_c$, $y_{i,(c_i,c_{-i}),c'_i}$ and $\be_{i,(c_i,c_{-i})}$ corresponding
to \eqref{e2}, \eqref{e3} and \eqref{e4} respectively. 
\begin{alignat}{1}
\max & \quad \qquad \sum_c \ga_c \tag{D} \label{dual} \\
\text{s.t.} &  
\sum_{i:c\in\bD_i\times\D_{-i}}\Bigl(\sum_{c'_i\in\bD_i}\bigl(c_i(\w)y_{i,(c_i,c_{-i}),c'_i}
- c'_i(\w)y_{i,(c'_i,c_{-i}),c_i}\bigr)+c_i(\w)\be_{i,c}\Bigr) \notag \\[-1.25ex]
& \qquad \qquad \qquad \qquad \quad +\kp\cdot{\textstyle \Pr_\D}(c)\pub(\w) 
\ \geq \ga_c \qquad \qquad \forall c\in{\bD}, \w\in\Om(c) \label{e5} \\
& \sum_{c'_i\in\bD_i}\bigl(y_{i,(c_i,c_{-i}),c'_i} - y_{i,(c'_i,c_{-i}),c_i}\bigr)+\be_{i,c_i,c_{-i}}
\leq {\textstyle \Pr_\D}(c) 
\qquad \forall i, c_i\in\bD_i, c_{-i}\in\D_{-i} \label{e6} \\[-1.5ex]
& \quad \qquad y,\be \geq 0. \label{e7}
\end{alignat}
With additive types, 
one can encode the LHS of \eqref{e5} as $\sum_i\tc_i(\w)$ for a suitably-defined additive
type (depending on $c$) $\tc_i=(\tc_{i,v})_{v\in\T_i}$.
Thus, the separation problem for constraints \eqref{e5} amounts to determining if the
optimal value of the \cm problem 
defined by a certain additive type profile, with possibly negative values, is at least
$\gm_c$.  
Hence, an optimal algorithm for the \cm problem can be used to solve \eqref{dual}, and
hence, \eqref{primal}, efficiently. 

\begin{theorem} \label{optcm}
With additive types, one can efficiently solve \eqref{primal} given an optimal algorithm
for the \cm problem. 
\end{theorem}

\begin{proof} 
Let $\A$ be an optimal algorithm for the \cm problem. 
Note that $\A$ is only required to work with nonnegative inputs.
We first observe that we can use $\A$ to find a solution that minimizes 
$\sum_ic_i(\w)+\kp\cdot\pub(\w)$ for any $\kp\geq 0$, even for an input
$c=\{c_{i,v}\}_{i,v\in\T_i}$ where some of the $c_{i,v}$s are negative.   
Let $A_i=\{v\in\T_i: c_{i,v}<0\}$. 
Clearly, if $\w^*$ is an optimal solution, then $A_i\sse\w^*_i$ (since $\pub(.)$ does not
increase upon adding covering objects). 
Define $c^+_{i,v}:=\max(0,c_{i,v})$ and $c_i^+:=\{c^+_{i,v}\}_{v\in\T_i}$.

Let $\Gm=\frac{1}{\kp}$ if $\kp>0$; otherwise let $\Gm=NU$, where
$U$ is a strict upper bound on $\max_{\w\in\Om}\pub(\w)$ and 
$N$ is an integer such that all the $c^+_{i,v}$s are integer multiples of $\frac{1}{N}$. 
Note that for any $\w,\w'\in\Om$, if $\sum_ic^+_i(\w)-\sum_ic^+_i(\w')$ is non-zero, then
its absolute value is at least $\frac{1}{N}$.
Also, $U$ and $N$ may be efficiently computed (for rational data) and
$\log(NU)$ is polynomially bounded.
Let $(S_1,\ldots,S_n)$ be the solution returned by $\A$ for the \cm problem on the input 
where all the $c^+_{i,v}$s are scaled by $\Gm$. The choice of $\Gm$ ensures that 
$$ 
\sum_i c_i^+(S_i)+\kp\cdot\pub\bigl(S_1,\ldots,S_n\bigr)
\leq\sum_i c_i^+(\w^*_i)+\kp\cdot\pub(\w^*)
=\sum_i\bigl(c_i(\w^*_i)-c_i(A_i)\bigr)+\kp\cdot\pub(\w^*). 
$$
(The first inequality clearly holds if $\kp>0$. If $\kp=0$ and
$\sum_i c_i^+(S_i)>\sum_i c_i^+(\w^*_i)$, then we have that
$\Gm\sum_i c_i^+(S_i)\geq\Gm\sum_i c_i^+(\w^*_i)+\frac{\Gm}{N}
>\Gm\sum_i c_i^+(\w^*_i)+\pub(\w^*)$, which contradicts the optimality of
$(S_1,\ldots,S_n)$ for the input $\{\Gm c^+_{i,v}\}_{i,v\in\T_i}$.)
So setting $\w'_i=A_i\cup S_i$ for every $i$ yields a feasible solution such that 
$\sum_ic_i(\w')+\kp\cdot\pub(\w')\leq\sum_i c_i(\w^*)+\kp\cdot\pub(\w^*)$; hence $\w'$ is 
an optimal solution.  

\medskip
Given a dual solution $(y,\beta,\gm)$, we can easily check if \eqref{e6}, \eqref{e7}
hold. Fix $c\in\bD$ and player $i$. 
Since we have additive types, if we define
$\tht^c_{i,v}=\sum_{c'_i\in\bD_i}\bigl(c_{i,v}y_{i,(c_i,c_{-i}),c'_i}-c'_{i,v}y_{i,(c'_i,c_{-i}),c_i}\bigr)+c_{i,v}\be_{i,c}$, 
then for every $\w\in\Om$, we can equate
$\sum_{c'_i\in\bD_i}\bigl(c_i(\w)y_{i,(c_i,c_{-i}),c'_i}-c'_i(\w)y_{i,(c'_i,c_{-i}),c_i}\bigr)+c_i(\w)\be_{i,c}$
with $\tht^c_i(\w):=\sum_{v\in\w_i}\tht^c_{i,v}$.

Let $I=\{i:c\in\bD_i\times\D_{-i}\}$. Constraints \eqref{e5} for $c$ can then be written
as $\min_{\w\in\Om(c)}\bigl(\sum_{i\in I}\tht^c(\w)+\kp\Pr_{\D}(c)\pub(\w)\bigr)\geq\gm_c$. 
Define $\tc$ as follows: 
$$
\tc_{i,v}=\begin{cases}
\gm_c+1 & \text{if $c_i=m_i\bon_{\T_i}$}, \\
\tht^c_{i,v} & \text{if $i\in I,\ c_i\in\D_i$}, \\ 
0 & \text{otherwise}. 
\end{cases}
$$
It is easy to see that \eqref{e5} holds for $c$ iff 
$\min_{\w\in\Om}\bigl(\sum_i\tc_i(\w)+\kp\cdot\Pr_{\D}(c)\pub(\w)\bigr)$%
---which can be computed using $\A$---is at least $\gm_c$. Thus, we can use the
ellipsoid method to solve \eqref{dual}. This also yields a compact dual consisting of
constraints \eqref{e6}, \eqref{e7} and the polynomially-many \eqref{e5} constraints that
were returned by the separation oracle during the execution of the ellipsoid method, whose
optimal value is $\OPT_{\text{\ref{dual}}}$. The dual of this compact dual is an LP of the
same form as \eqref{primal} but with polynomially many $x_{c,\w}$-variables; solving this
yields an optimal solution to \eqref{primal}.
\end{proof}

\noindent
Complementing Theorem~\ref{optcm}, we argue that a feasible solution $(x,p)$ to
\eqref{primal} can be extended to a support-based-(IC-in-expectation, IR) mechanism
having expected disutility at most the value of $(x,p)$ (Theorem~\ref{multiround}).
Combining this with Theorem~\ref{optcm} yields the corollary that an optimal algorithm for 
the \cm problem can be used to obtain an optimal mechanism for the \paym problem
(Corollary~\ref{bicredn1}).  

\begin{theorem} \label{multiround}
We can extend a feasible solution $(x,p)$ to \eqref{primal} 
to a support-based-(IC-in-expectation, IR) mechanism with expected disutility
$\sum_{c}\Pr_\D(c)\bigl(\sum_i p_{i,c}+\kp\sum_{\w}x_{c,\w}\pub(\w)\bigr)$. 
\end{theorem}

\begin{proof}
Let $\Om'=\{\w: x_{c,\w}>0\text{ for some }c\in\bD\}$.
We use $x_c$ to denote the vector $\{x_{c,\w}\}_{\w\in\Om'}$.
Consider a player $i$, $c_{-i}\in\D_{-i}$, and $c_i,c'_i\in\bD_i$. 
Note that \eqref{e3} implies that if
$x_{(c_i,c_{-i})}=x_{(c'_i,c_{-i})}$, then $p_{i,(c_i,c_{-i})}=p_{i,(c'_i,c_{-i})}$.
For $c_{-i}\in\D_{-i}$, define 
$\range(i,c_{-i})=\bigl\{x_{(c_i,c_{-i})}: (c_i,c_{-i})\in\bD\bigr\}$, and for
$y=x_{(c_i,c_{-i})}\in\range(i,c_{-i})$ define $p_{i,y}$ to be $p_{i,(c_i,c_{-i})}$ (which
is well defined by the above argument).

We now define the randomized mechanism $M=\bigl(\A,\{q_i\}\bigr)$, where $\A(c)$ and
$q_i(c)$ denote respectively the probability distribution over allocations and the
expected payment to player $i$, on input $c$. We sometimes view $\A(c)$ equivalently as
the random variable specifying the allocation chosen for input $c$.
Fix an allocation $\w_0\in\Om$. 
Consider an input $c$. If $c\in\bD$, we set $\A(c)=x_c$, and $q_i(c)=p_{i,c}$ for all
$i$. 
So consider $c\notin\bD$. 
If there is no $i$ such that $c_{-i}\in\D_{-i}$,
we simply set $\A(c)=\w_0$, $q_i(c)=c_i(\w_0)$ for all $i$; such a $c$ does not figure in
the support-based (IC, IR) conditions. 
Otherwise there is a unique $i$ such that 
$c_{-i}\in\D_{-i},\ c_i\in C_i\sm\bD_i$. Set
$\A(c)=\arg\max_{y\in\range(i,c_{-i})}\bigl(p_{i,y}-\sum_{\w\in\Om'}c_i(\w)y_\w\bigr)$ and 
$q_j(c)=p_{j,\A(c)}$ for all players $j$. Note that $(c_i,c_{-i})$ figures in \eqref{e3}
{\em only} for player $i$.
Crucially, note that since $y=x_{(m_i,c_{-i})}\in\range(i,c_{-i})$ and
$\sum_{\w\in\Om}c_i(\w)y_\w=0$ by definition, 
we always have $q_i(c)-\E_\A[c_i(\A(c))]\geq 0$.
Thus, by definition, and by \eqref{e3}, we have ensured that $M$ is 
support-based (IC, IR)-in-expectation and its expected disutility is exactly the value of
$(x,p)$. 
This can be modified so that IR holds with probability 1.  

The above procedure is efficient if 
$\sum_{\w\in\Om'}c_i(\w)x_{c,\w}$ can be calculated efficiently. This is clearly true if
$|\Om'|$ is polynomially bounded, but it could hold under weaker conditions as well.
\end{proof}

\begin{corollary} \label{bicredn1}
Given an optimal algorithm for the \cm problem, we can efficiently obtain an optimal
support-based-(IC-in-expectation, IR) mechanism for the \paym problem in multidimensional
settings with additive types. 
\end{corollary}

\vspace{-1ex}
\paragraph{Using approximation algorithms.}
The \cm problem is however often \nphard (e.g., for vertex cover), and we would like to
be able to exploit approximation algorithms for the \cm problem 
to obtain near-optimal mechanisms.  
The usual approach is to use an approximation algorithm 
to ``approximately'' separate over constraints \eqref{e5}.%
\footnote{For revenue-maximization problems in packing domains, this simple approach 
{\em does} indeed work in single-dimensional settings and settings with additive types. 
This is because the dual separation problem is now an SWM problem, and 
it is easy to use a $\rho$-approximation algorithm for the SWM problem 
for nonnegative inputs to obtain a $\rho$-approximate 
solution with {\em arbitrary}, positive or negative, inputs.
This yields a simple extension of some of the results in~\cite{DobzinskiFK11}; see
Appendix~\ref{append-packextn}.}
However, this does not work here since the \cm problem that one needs to solve 
in the separation problem involves negative costs, which renders
the usual notion of approximation meaningless.
Instead, if the \cm problem admits a certain type of LP-relaxation \eqref{cmlp}, then we 
argue that one can solve a relaxation of \eqref{primal} where the allocation-set is
the set of extreme points of \eqref{cmlp} (Theorem~\ref{optnewlp}).  
For single-dimensional problems (Section~\ref{single}), we leverage this to obtain strong
and far-reaching results. 
We show that a $\rho$-approximation algorithm relative to
\eqref{cmlp} 
can be used to ``round'' the optimal solution to this relaxation to a
support-based-(IC-in-expectation, IR)-mechanism losing a $\rho$-factor in the disutility
(Theorem~\ref{bicredn2}). Thus, we obtain near-optimal mechanisms for a variety of
single-dimensional problems.  

Suppose that the \cm problem admits an LP-relaxation of the following form, where
$c=\{c_{i,v}\}_{i\in[n],v\in\T_i}$ is the input type-profile. 
\begin{equation}
\min \quad c^Tx+d^Tz \qquad
\text{s.t.} \qquad Ax+Bz\geq b, \quad x, z\geq 0. 
\tag{C-P} \label{cmlp}
\end{equation}
Intuitively $x$ encodes the allocation chosen, 
and $d^Tz$ encodes $\pub(.)$. 
For $x\geq 0$, define 
$z(x):=\arg\min\{d^Tz: (x,z)\text{ is feasible to \eqref{cmlp}}\}$; if there is no $z$
such that $(x,z)$ is feasible to \eqref{cmlp}, set $z(x):=\bot$.
Define $\Om_\lp:=\{x: z(x)\neq\bot,\ \ 0\leq x_{i,v}\leq 1\ \forall i, v\in\T_i\}$. 
We require that: 
(a) a $\{0,1\}$-vector $x$ is in $\Om_\lp$ iff it is the characteristic vector of an
allocation $\w\in\Om$, and in this case, we have $d^Tz(x)=\pub(\w)$;
(b) $A\geq 0$;  
(c) for any input $c\geq 0$ to the covering problem, \eqref{cmlp} is not unbounded, and if it
has an optimal solution, it has one where $x\in\Om_\lp$;
(d) for any $c$, we can efficiently find an optimal solution to \eqref{cmlp} or detect
that it is unbounded or infeasible.

\newcommand{\optcmlp}{\ensuremath{\OPT_{\text{\ref{cmlp}}}}\xspace}

We extend the type $c_i$ of each player $i$ and $\pub$ to assign values also to points in 
$\Om_\lp$: define $c_i(x)=\sum_{v\in\T_i}c_{i,v}x_{i,v}$ and $\pub(x)=d^Tz(x)$ for
$x\in\Om_\lp$. 
Let $\Om_\ex$ denote the finite set of extreme points of $\Om_\lp$. Condition (a)
ensures that $\Om_\ex$ contains the characteristic vectors of all feasible allocations.
Let \newlp denote the relaxation of \eqref{primal}, where we replace the set of feasible
allocations $\Om$ with $\Om_\ext$ (so $\w$ indexes $\Om_\ext$ now), and for $c\in\bD$ with
$c_i=m_i\bon(\T_i)$, we now define $\Om(c):=\{\w\in\Om_\ext:\w_{i,v}=0\ \forall v\in\T_i\}$. 
Since one can optimize efficiently over $\Om_\lp$, and hence $\Om_\ext$, even for negative
type-profiles, we have the following.

\begin{theorem} \label{optnewlp}
We can efficiently compute an optimal solution to \newlp.
\end{theorem}

\section{Single-dimensional problems} \label{single}
Corollary~\ref{bicredn1} immediately yields results for certain single-dimensional
problems (see Table~\ref{onedapps}), most notably, {\em single-item procurement auctions}.
We now substantially expand the scope of \paym problems for which one can
obtain near-optimal mechanisms by showing how to leverage ``LP-relative'' approximation 
algorithms for the \cm problem. 
(As noted earlier, and sketched in Appendix~\ref{append-packextn}, a simpler approach can
be used to leverage approximation algorithms for revenue-maximization in packing domains.) 
Suppose that the \cm problem can be encoded as \eqref{cmlp}. 
An {\em LP-relative $\rho$-approximation algorithm} for the \cm problem is a polytime
algorithm that for any input $c\geq 0$ to the covering problem, returns 
a $\{0,1\}$-vector $x\in\Om_\lp$ 
such that $c^Tx+d^Tz(x)\leq\rho\optcmlp$.   
Using the convex-decomposition procedure in~\cite{LaviS11} (see Section 5.1
of~\cite{LaviS11}), one can show the following; the proof appears at the end of this
section. 

\begin{lemma} \label{cvxdec1}
Let $x\in\Om_\lp$. Given an LP-relative $\rho$-approximation algorithm $\A$ for the \cm 
problem, one can efficiently obtain $(\ld^{(1)}, x^{(1)}),\ldots,
(\ld^{(k)}, x^{(k)})$, where $\sum_\ell \ld^{(\ell)}=1, \ld\geq 0$, and 
$x^{(\ell)}$ is a $\{0,1\}$-vector in $\Om_\lp$ for all $\ell$, such that 
$\sum_\ell\ld^{(\ell)} x^{(\ell)}_{i,v}=\min(\rho x_{i,v},1)$ for all $i,v\in\T_i$, and
$\sum_\ell\ld^{(\ell)} d^Tz(x^{(\ell)})\leq\rho d^Tz(x)$.  
\end{lemma}

\begin{theorem} \label{bicredn2}
Given an LP-relative $\rho$-approximation algorithm for the \cm problem, one can obtain a 
polytime $\rho$-approximation support-based-(IC-in-expectation, IR) mechanism for the \paym
problem. 
\end{theorem}

\begin{proof}
We solve \newlp to obtain an optimal solution $(x,p)$. 
Since $|\T_i|=1$ for all $i$, it will be convenient to view $\w\in\Om_\lp$ as a vector
$\{\w_i\}_{i\in[n]}\in[0,1]^n$, where $\w_i\equiv\w_{i,v}$ for the single covering object  
$v\in\T_i$. 
Fix $c\in\bD$. Define $y_c=\sum_{\w\in\Om_\ext}x_{c,\w}\w$ (which can be efficiently
computed since $x$ has polynomial support). 
Then, $\sum_{\w\in\Om_\ext} c_i(\w)x_{c,\w}=c_iy_{c,i}$ and
$d^Tz(y)\leq\sum_{\w\in\Om_\ext}\pub(\w)x_{c,\w}$.  
By Lemma~\ref{cvxdec1}, 
we can efficiently find a point $\ty_c=\sum_{\w\in\Om}\tx_{c,\w}\w$, where 
$\tx_c\geq 0, \sum_{w\in\Om} \tx_{c,\w}=1$, in the convex hull 
of the $\{0,1\}$-vectors in $\Om_{\lp}$ 
such that $\ty_{c,i}=\min(\rho y_{c,i},1)$ 
for all $i$, and $\sum_{w\in\Om}\tx_{c,\w}\pub(\w)\leq\rho d^Tz(y)$.

We now argue that one can obtain payments $\{q_{i,c}\}$ such that $(\tx,q)$ is feasible to
\eqref{primal} and $q_{i,c}\leq\rho p_{i,c}$ for all $i, c\in\bD$. 
Thus, the value of $(\tx,q)$ is at most $\rho$ times the value of $(x,p)$.
Applying Theorem~\ref{multiround} to $(\tx,q)$ yields the desired result.

Fix $i$ and $c_{-i}\in\D_{-i}$. 
Constraints \eqref{e45} and \eqref{e3} ensure that $y_{(m_i,c_{-i}),i}=0$, and
$y_{(c_i,c_{-i}),i}\geq y_{(c'_i,c_{-i}),i}$ for all $c_i,c'_i\in\bD_i$ s.t. $c_i<c'_i$.
Hence, $\ty_{(m_i,c_{-i}),i}=0,\ \ty_{(c_i,c_{-i}),i}\geq\ty_{(c'_i,c_{-i}),i}$ for
$c_i,c'_i\in\bD_i,\ c_i>c'_i$.
Define $q_{i,(m_i,c_{-i})}=0$. Let $0\leq c_i^1<c_i^2<\ldots<c_i^{k_i}$ be the values in $\D_i$.
For $c_i=c_i^\ell$, 
define 
$$
q_{i,(c_i,c_{-i})}=
c_{i}\ty_{(c_i,c_{-i}),i}+\sum_{t=\ell+1}^{k_i}(c_i^t-c_i^{t-1})\ty_{(c_i^t,c_{-i}),i}.
$$
Since $\sum_{\w\in\Om}c_i(\w)\tx_{(c_i,c_{-i}),\w}=c_i\ty_{(c_i,c_{-i}),i}$, 
\eqref{e4} holds. By construction, for consecutive values
$c_i=c_i^\ell,\ c'_i=c_i^{\ell+1}$, we have 
$q_{i,(c_i,c_{-i})}-q_{i,(c'_i,c_{-i})}
=c_i\bigl(\ty_{(c_i,c_{-i}),i}-\ty_{(c'_i,c_{-i}),i}\bigr)$, which is at most
$$ 
\rho\cdot c_i\bigl(y_{(c_i,c_{-i}),i}-y_{(c'_i,c_{-i}),i}\bigr)
\leq\rho\bigl(p_{i,(c_i,c_{-i})}-p_{i,(c'_i,c_{-i})}\bigr).
$$
Since $q_{i,(m_i,c_{-i})}=0\leq\rho p_{i,(m_i,c_{-i})}$, this implies that
$q_{i,(c_i,c_{-i})}\leq\rho p_{i,(c_i,c_{-i})}$. 
Finally, it is easy to verify that for any $c_i,c'_i\in\D_i$, we have
$q_{i,(c_i,c_{-i})}-q_{i,(c'_i,c_{-i})}
\geq c_i\bigl(\ty_{(c_i,c_{-i}),i}-\ty_{(c'_i,c_{-i}),i}\bigr)$, so $(\tx,q)$
satisfies \eqref{e3}.
\end{proof}

Corollary~\ref{bicredn1} and Theorem~\ref{bicredn2} yield polytime near-optimal
mechanisms for a host of single-dimensional \paym problems.
Table~\ref{onedapps} summarizes a few applications. 
Even for single-item procurement auctions, these are the {\em first} results for   
\paym problems with correlated players satisfying a notion stronger than (BIC, interim
IR). 

\begin{table}[ht!]
\begin{center}
\hspace*{-0.25in}
\begin{tabular}{|p{3.7in}|p{1.2in}|p{1in}|} \hline
{\centerline{\bf Problem}} & 
{\centerline{\bf Approximation}} & 
{\centerline{\bf Due to}} \\[-2ex] \hline
{{\em Single-item procurement auction}: buy one item provided by $n$ players} & \ 1 &
{\ Corollary~\ref{bicredn1}} \\ \hline
{{\em Metric \ufl}: players are facilities, output should be a \ufl solution} & \ 
{1.488 using~\cite{Li13}} &
{\ Theorem~\ref{bicredn2}} \\ \hline
{{\em Vertex cover}: players are nodes, output should be a vertex cover} & \ 2 &
{\ Theorem~\ref{bicredn2}} \\ \hline
{{\em Set cover}: players are sets, output should be a set cover} & \ $O(\log n)$ &
{\ Theorem~\ref{bicredn2}} \\ \hline
{{\em Steiner forest}: players are edges, output should be a Steiner forest} & \ 2 &
{\ Theorem~\ref{bicredn2}} \\ \hline 
{{\em Multiway cut} (a), {\em  Multicut} (b): players are edges, output should be a
multiway cut in (a), or a multicut in (b)} & 
{\ 2 for (a) \newline \hspace*{0.4ex} $O(\log n)$ for (b)} & 
{\ Theorem~\ref{bicredn2}} \\ \hline
\end{tabular}
\end{center}
\vspace{-1ex} 
\caption{\label{onedapps} Results for some representative single-dimensional
  \paym problems.} 
\end{table}

\begin{proofof}{Lemma~\ref{cvxdec1}}
It suffices to show that the LP \eqref{cvxprim} can be solved in polytime and its optimal
value is $1$. Throughout, we use $\ell$ to index $\{0,1\}$ vectors in $\Om_\lp$. (Recall
that these correspond to feasible allocations.)

\vspace{-3ex}

{\centering\small
\noindent \hspace*{-5ex}
\begin{minipage}[t]{.52\textwidth}
\begin{alignat}{2}
\max & \quad & \sum_\ell\lambda^{(\ell)} & \tag{Q} \label{cvxprim} \\
\text{s.t.} && \sum_\ell\lambda^{(\ell)} x^{(\ell)}_{i,v} = 
\min(&\rho x_{i,v},1) \quad \forall i,v\in\T_i \label{1} \\
&& \sum_{\ell}\lambda^{(\ell)}d^T\bigl(z(x^{(\ell)})\bigr) & 
\leq \rho d^Tz(x) \label{2} \\
&& \sum_\ell\ld^{(\ell)} & \leq 1 \label{3} \\ 
&& \lambda & \geq 0. \notag
\end{alignat}
\end{minipage}
\ \rule[-31ex]{1pt}{28ex}\ 
\begin{minipage}[t]{.51\textwidth}
\begin{alignat}{1}
\min & \quad \sum_{i,v\in\T_i}\min(\rho x_{i,v},1)\alpha_{i,v}
+\rho d^Tz(x)\cdot\beta + \tht \tag{R} \label{cvxdual} \\ 
\text{s.t.} & \quad \sum_{i,v\in\T_i}x^{(\ell)}_{i,v}\alpha_{i,v}
+d^T\bigl(z(x^{(\ell)})\bigr)\beta + \tht \geq 1 \quad \forall \ell \label{4} \\ 
& \qquad \beta,\tht \geq 0. \notag
\end{alignat}
\end{minipage}
}

\smallskip
Here the $\alpha_\ell$s, $\beta$ and $\tht$ are the dual variables corresponding to  
constraints \eqref{1}, \eqref{2}, and \eqref{3} respectively.
Clearly, $\OPT_{\eqref{cvxdual}}\leq 1$ since $\tht=1$, $\al_{i,v}=0=\beta$ for all
$i,v$ is a feasible dual solution. 

Suppose $(\hal,\hbeta,\htht)$ is a feasible dual solution of value less than 1.
Set $\tal_{i,v}=\hal_{i,v}$ if $\hal_{i,v}\geq 0$ and $\rho x_{i,v}\leq 1$, and
$\tal_{i,v}=0$ otherwise. 
Let $\Gm=\frac{1}{\hbeta}$ if $\hbeta>0$ and equal to $2N d^Tz$ otherwise, where $N$ is is
such that for all $\{0,1\}$-vectors $x^{(\ell)}\in\Om_\lp$, we have that $c^Tx^{(\ell)}>c^Tx$
implies $c^Tx^{(\ell)}\geq c^Tx+\frac{1}{N}$. Note that we can choose $N$ so that its size
is $\poly(\I,\text{size of $x$})$.
Consider the \cm problem defined by the input $\Gm\tal$. Running $\A$ on this input, we
obtain a $\{0,1\}$-vector $x^{(\ell)}\in\Om_\lp$ whose total cost is at most $\rho$ times
the cost of the fractional solution $\bigl(x,z(x)\bigr)$. This translates to
\begin{equation}
\sum_{i,v}x^{(\ell)}_{i,v}\tal_{i,v}+d^T\bigl(z(x^{(\ell)})\bigr)\hbeta
\leq \rho\Bigl(\sum_{i,v}x_{i,v}\tal_{i,v}+d^Tz(x)\cdot\hbeta\Bigr).
\label{5}
\end{equation}
Now augment $x^{(\ell)}$ to the following $\{0,1\}$-vector $\hx$: set $\hx_{i,v}=1$ if
$\rho x_{i,v}>1$ or $\hal_{i,v}<0$, and $x^{(\ell)}_{i,v}$ otherwise. Then $\hx$ is the
characteristic vector of a feasible allocation, since we have only added covering objects
to the allocation corresponding to $x^{(\ell)}$; hence $\hx\in\Om_\lp$. We have
$d^Tz(\hx)=\pub(\hx)\leq\pub(x^{(\ell)})=d^T\bigl(z(x^{\ell})\bigr)$ and
$$
\sum_{i,v}\hx_{i,v}\hal_{i,v}
=\sum_{\substack{i,v:\rho x_{i,v}>1 \\ \text{or }\hal_{i,v}<0}}\hal_{i,v}+\sum_{i,v}x^{(\ell)}_{i,v}\tal_{i,v}
\leq\sum_{\substack{i,v:\rho x_{i,v}>1 \\ \text{or }\hal_{i,v}<0}}\min(\rho x_{i,v},1)\hal_{i,v}
+\sum_{i,v}x^{(\ell)}_{i,v}\tal_{i,v}.
$$
Combined with \eqref{5}, this shows that 
\begin{equation*}
\begin{split}
\sum_{i,v}\hx_{i,v}\hal_{i,v}+d^Tz(\hx)\hbeta
& \leq\sum_{\substack{i,v:\rho x_{i,v}>1 \\ \text{or }\hal_{i,v}<0}}\min(\rho x_{i,v},1)\hal_{i,v}
+\sum_{i,v:\tal_{i,v}>0}\rho x_{i,v}\tal_{i,v}+\rho d^Tz(x)\cdot\hbeta \\
& =\sum_{i,v}\min(\rho x_{i,v},1)\hal_{i,v}+\rho d^Tz(x)\cdot\hbeta < 1-\htht
\end{split}
\end{equation*}
which contradicts that $(\hal,\hbeta,\htht)$ is feasible to \eqref{cvxdual}. Hence,
$\OPT_{\eqref{cvxprim}}=\OPT_{\eqref{cvxdual}}=1$. 

Thus, we can add the constraint 
$\sum_{i,v\in\T_i}\min(\rho x_{i,v},1)\alpha_{i,v}+\rho d^Tz(x)\cdot\beta+\tht\leq 1$ to
\eqref{cvxdual} without altering anything. If we solve the resulting LP using the ellipsoid
method, and take the inequalities corresponding to the violated inequalities \eqref{4}
found by $\A$ during the ellipsoid method, then we obtain a compact LP with only a
polynomial number of constraints that is equivalent to \eqref{cvxdual}. The dual of this
compact LP yields an LP equivalent to \eqref{cvxprim} with a polynomial number of
$\ld^{(\ell)}$ variables which we can solve to obtain the desired convex decomposition. 
\end{proofof}

\section{Multidimensional problems} \label{multi}

We obtain results for multidimensional \paym problems via two distinct approaches. One is
by directly applying Corollary~\ref{bicredn1} (e.g., Theorem~\ref{multiproc}). The other 
approach is based on again moving to an LP-relaxation of the \cm problem and utilizing
Theorem~\ref{optnewlp} in conjunction with a {\em stronger} LP-rounding approach. This
yields results for multidimensional (metric) \ufl and its variants (Theorem~\ref{uflres}).

\vspace{-1ex}
\paragraph{Multi-item procurement auctions.} 
Here, we have $n$ sellers and $k$ (heterogeneous) items. 
Each seller $i$ has a {\em supply vector} $s_i\in\Z_+^k$ denoting his supply for the
various items, and the buyer has a {\em demand vector} $d\in\Z_+^k$ specifying his
demand for the various items. This is public knowledge.
Each seller $i$ has a {\em private} cost-vector $c_i\in\R_+^k$, where $c_{i,\ell}$ is the
cost he incurs for supplying {\em one unit} of item $\ell$. 
A feasible solution is an allocation specifying how many units of each item each seller
supplies to the buyer such that 
for each item $\ell$, each seller $i$ provides at most $s_{i,\ell}$ units of $\ell$ and
the buyer obtains $d_\ell$ total units of $\ell$.
The corresponding \cm problem is a min-cost flow problem (in a bipartite graph), which
can be efficiently solved optimally, thus we obtain a polytime optimal mechanism.

\begin{theorem} \label{multiproc}
There is a polytime optimal support-based-(IC-in-expectation, IR) mechanism for
multi-unit procurement auctions.
\end{theorem}

\vspace{-1ex}
\paragraph{Multidimensional budgeted (metric) uncapacitated facility location (\ufl).} 
Here, we have a set $\mE$ of clients that need to be serviced by facilities, and a set
$\F$ of locations where facilities may be opened. Each player $i$ may provide facilities
at the locations in $\T_i\sse\F$. We may assume that the $\T_i$s are disjoint. 
For each facility $\ell\in\T_i$ that is opened, $i$ incurs a private opening cost 
$f_{\ell}\equiv f_{i,\ell}$, and assigning client $j$ to an open facility $\ell$ incurs a  
publicly-known assignment cost $d_{\ell j}$, where the $d_{\ell j}$s form a
metric. We are also given a public assignment-cost budget $B$. 
The goal in \bufl is to open a subset $F\sse\F$ of facilities and assign each client $j$ 
to an open facility $\sg(j)\in F$ so as to minimize 
$\sum_{\ell\in F}f_{\ell}+\sum_{j\in \mE}d_{\sg(j) j}$ subject to 
$\sum_{j\in \mE}d_{\sg(j) j}\leq B$;
\ufl is the special case where $B=\infty$.
We can define $\pub(T_1,\ldots,T_n)$ to be the total assignment cost if this is at most
$B$, and $\infty$ otherwise. 

Let $\iopt$ denote the expected disutility of an optimal mechanism for \bufl.
We obtain a mechanism with expected disutility at most $2\iopt$ that 
always returns a solution with expected assignment cost at most $2B$.
Consider the following LP-relaxation for \bufl. 
\begin{gather*}
\min \ \sum_{\ell\in\F} f_\ell x_\ell+\sum_{j\in\mE, \ell\in\F}d_{\ell j}z_{\ell j}\tag{BFL-P} 
\label{bfl-p} 
\quad \text{s.t.} \\
\sum_{j\in\mE, \ell\in\F}d_{\ell j}z_{\ell j}\leq B, \quad \
\sum_{\ell\in\F} z_{\ell j} \geq 1 \ \ \forall j\in\mE, \quad \
0\leq z_{\ell j} \leq x_\ell \ \ \forall \ell\in\F, j\in\mE. 
\end{gather*}

Let $\flp$ denote \eqref{bfl-p} with $B=\infty$, and $\flopt$ denote its optimal value. 
We say that an algorithm $\A$ is a {\em Lagrangian multiplier preserving} (LMP)
$\rho$-approximation algorithm for \ufl if for every instance, it returns a solution 
$(F,\sg)$ such that 
$\rho\sum_{\ell\in F}f_\ell+\sum_{j\in\mE}d_{\sg(j)j}\leq\rho\cdot\flopt$.
In~\cite{MinooeiS12}, it is shown that given such an algorithm $\A$, one can take any 
solution $(x,z)$ to \flp and obtain a convex combination of {\em \ufl solutions} 
$(\ld^{(1)}; F^{(1)},\sg^{(1)}),\ldots,(\ld^{(k)}; F^{(k)},\sg^{(k)})$---so 
$\ld\geq 0,\ \sum_r\ld^{(r)}=1$---such that $\sum_{r:\ell\in F^{(r)}}\ld^{(r)}=x_\ell$ for
all $\ell$ and $\sum_r\ld^{(r)}\bigl(\sum_j d_{\sg^{(r)}(j)j}\bigr)\leq\rho\sum_{j,\ell} d_{\ell j}z_{\ell j}$. 
An LMP 2-approximation algorithm for \ufl is known~\cite{Jain}.

\begin{lemma} \label{uflround}
Given an LMP $\rho$-approximation algorithm for \ufl, one can design a polytime
support-based-(IC-in-expectation, IR) mechanism for \bufl whose expected disutility is at 
most $\rho\iopt$ while violating the budget by at most a $\rho$-factor.
\end{lemma}

\begin{proof} 
The LP-relaxation \eqref{bfl-p} for the \cm problem is of the form \eqref{cmlp} and
satisfies the required properties.  
Recall that for $x\geq 0$, $z(x)$ denotes the min-cost completion of $x$ to a feasible
solution to \eqref{bfl-p} if one exists, and is $\bot$ if there is no such completion of
$x$. Let $\Om_\lp:=\{x: z(x)\neq\bot,\ \ 0\leq x_\ell\leq 1\ \forall\ell\}$. 
For integral $\w\in\Om_\lp$, $z(\w)$ specifies the assignment where each client $j$ is
assigned to the nearest open facility. 
By Theorem~\ref{optnewlp}, one can efficiently compute an optimal solution $(X,p)$ to the
relaxation of \eqref{primal} where the set of feasible allocations is the set $\Om_\ex$ of
extreme points of $\Om_\lp$.

We round $(X,p)$ to a feasible solution to \eqref{primal} by proceeding as in the
proof of Theorem~\ref{bicredn2}. Let $\Om_\ufl$ be the set of characteristic vectors of
open facilities of all integral \ufl solutions. 
We use $\ell$ to index facilities in $\F$ and $j$ to index clients in $\mE$.  
Fix $c\in\bD$. Define $y_c=\sum_{\w\in\Om_\ex}X_{c,\w}\w$, so
$\sum_{w\in\Om_\ext}c_i(\w)X_{c,\w}=\sum_{\ell\in\T_i}f_\ell y_{c,\ell}$.
Let $z_c=\sum_{\w\in\Om_\ex}X_{c,\w}z(\w)$, so $\sum_{j,\ell}z_{c,\ell j}d_{\ell j}\leq B$.
We use the LMP $\rho$-approximation algorithm to express
$y_c$ as a convex combination $\sum_{\w\in\Om_\ufl}\tx_{c,\w}\w$ of (integral)
\ufl-solutions such that the expected assignment cost
$\sum_{\w\in\Om_\ufl}\tx_{c,\w}\sum_{j,\ell}z(\w)_{\ell j}d_{\ell j}$ is at most
$\rho\sum_{j,\ell}d_{\ell j}z_{c,\ell j}\leq\rho B$.
Hence, $(\tx,p)$ is a feasible solution to \eqref{primal}. Theorem~\ref{multiround} now 
yields the desired mechanism. 
\end{proof}

\begin{theorem} \label{uflres}
There is a polytime support-based-(IC-in-expectation, IR) mechanism for \bufl
with expected disutility at most $2\iopt$, which always returns a solution with
expected assignment cost at most $2B$. 
\end{theorem}

\section{Extensions: alternative solution concepts} \label{extn}
We now investigate the \paym problem under various
alternative solution concepts. In Section~\ref{weaker}, we consider solution concepts
weaker than support-based (IC, IR), but yet robust enough to ensure that truthful
participation is an ex-post no-regret choice for every player at every type profile in the 
support of the underlying distribution. We show that {\em all} our guarantees extend
readily to these solution concepts. (Note that a weaker solution concept does not
necessarily mean that the corresponding \paym mechanism-design problem is a simpler
problem; a weaker solution concept enlarges the space of allowed mechanisms, which could
make it more- or less- difficult to search for an optimal mechanism.) 
In Section~\ref{dsic}, we consider the stronger solution concept
of (DSIC (in-expectation, IR), and obtain results for single-dimensional settings but at
the expense of increasing the running time to exponential in the number of players.

\subsection{Solution concepts weaker than support-based (IC, IR)} \label{weaker}
Consider the following weakenings of support-based (IC, IR) (Definition~\ref{expostbic}). 

\medskip
\noindent For every player $i$, 

\vspace*{-4.5ex}
\begin{alignat}{3}
u_i(\bc_i,c_{-i};\bc_i) & \geq u_i(c_i,c_{-i};\bc_i), \quad 
u_i(\bc_i,c_{-i};\bc_i) && \geq 0, \qquad 
&& \text{for all $(\bc_i,c_{-i})\in\D$,\ \ $c_i: (c_i,c_{-i})\in\D$} 
\label{opta} \\
u_i(\bc_i,c_{-i};\bc_i) & \geq u_i(c_i,c_{-i};\bc_i), \quad 
u_i(\bc_i,c_{-i};\bc_i) && \geq 0, \qquad 
&& \text{for all $(\bc_i,c_{-i})\in\D$,\ \ $c_i\in C_i$} \label{optb} \\
u_i(\bc_i,c_{-i};\bc_i) & \geq u_i(c_i,c_{-i};\bc_i), \quad 
u_i(\bc_i,c_{-i};\bc_i) && \geq 0, \qquad 
&& \text{for all $\bc_i,c_i\in\D_i$,\ \ $c_{-i}\in\D_{-i}$} \label{optc}
\end{alignat}

All three solution concepts, \eqref{opta}--\eqref{optc}, ensure that truthful
participation is in the best interest of every player $i$ at every type-profile in $\D$
even at the ex-post stage when he knows the realized types of all players, but for
varying choices of lies: in \eqref{optb}, the lie could be anything, so a mechanism
satisfying \eqref{optb} is (BIC, interim IR) for every distribution whose support is a
subset of $\D$;
in \eqref{optc}, the ``best interest'' is among lies consistent with $i$'s
support; and in \eqref{opta}, the ``best interest'' is among lies consistent
with the support of the distribution.

We now argue that our results extend to these notions. For notions \eqref{opta} and
\eqref{optc}, one can simply incorporate all the IC and IR constraints in the LP. Note 
that there are $O(n|\D|^2)$ such constraints under \eqref{opta}, and $O(n|\D|^3)$
constraints under \eqref{optc}, so the size of the resulting LP is $\poly(n,|\D|)$. 
Theorem~\ref{optcm} continues to hold for the resulting LP, due to the same arguments.
For both notions, an LP solution immediately yields a randomized mechanism satisfying
that notion, except that utility is replaced by expected utility; for type profiles not
included in the LP, we may output any outcome $\w_0$ (and any prices).
The refinements for single-dimensional settings and multidimensional FL work in the same
fashion as before: the appropriate LP (e.g., \newlp) is modified to include the
appropriate set of IC and IR constraints and solved as before. The rounding of an LP
solution to obtain a suitable mechanism proceeds as before.
Consequently, all of our results extend to these two notions with minimal effort.

For notion \eqref{optb}, we incorporate constraints \eqref{optb} but restrict $c_i$ to lie  
in $\bD_i$. Recall that $\bD_i=\D_i\cup\{m_i\bon\T_i\}$. We also include constraints
\eqref{e45} as before. Again, the resulting LP can be solved given an optimal algorithm
for the \cm problem. The LP solution can be extended to a mechanism 
{\em exactly} as in Theorem~\ref{multiround}, and it is easy to see
that this extension satisfies \eqref{optb}. Therefore, all our results extend to this
notion as well.

\subsection{Dominant-strategy IC mechanisms} 
\label{dsic}
We can strengthen our results from Section~\ref{single} to obtain (near-) optimal 
{\em dominant-strategy incentive compatible} (DSIC) mechanisms for single-dimensional
problems in time exponential in $n$. 
Thus, we obtain polytime mechanisms for any constant number of players. 

The key change is in the LP \eqref{primal} (or \newlp), 
where we now enforce \eqref{e2}--\eqref{e45} for every player $i$ and every type profile 
in $\prod_i\bD_i$. 
(Note that, as before, we can only enforce IC and IR constraints for a finite set of
type profiles.) 
Theorem~\ref{optcm}, as also the rounding procedure and arguments in
Theorem~\ref{bicredn2} proceed essentially identically to yield a near-optimal solution to
this LP. 
We prove that in single-dimensional settings, enforcing the IC, IR constraints for the set
$\prod_i\bD_i$ of type profiles enables one to extend the LP solution to a
(DSIC-in-expectation, IR) mechanism without increasing the expected disutility. Thus, we
obtain the same guarantees as in Table~\ref{onedapps}, but under the stronger solution
concept of (DSIC-in-expectation, IR). 

We focus on single-dimensional settings here because at
various places, our arguments rely on the well-known equivalence between monotonic
allocation rules and DSIC-implementable allocation rules. 
We do not know if a similar result holds for multidimensional settings. This, and various
other unanswered questions emerge from our result; 
we mention a few of these below, before delving into our construction for
single-dimensional settings.
\begin{list}{(\alph{enumi})}{\usecounter{enumi} \topsep=0ex \itemsep=0ex
    \addtolength{\leftmargin}{-1ex}} 
\item In multidimensional settings, what finite subsets $C'\sse C$ of the type space have
the property that enforcing the IC and IR constraints for every player $i$ and every type
profile in $C'$ suffice to extend an LP solution to a (DSIC-in-expectation, IR mechanism)?
\item Does $C'=\prod_i\bD_i$ have this extension property (as is the case in
single-dimensional settings)? 
\item Is there some $C'$ of size $\poly(n,|D|)$ with this extension property?
\end{list}

\medskip
We now describe briefly the changes required to obtain (DSIC-in-expectation, IR)
mechanisms. 
Analogous to Lemma~\ref{estim}, we can obtain estimates $m_i$ such that there is an
optimal mechanism $M^*$ such that on any input $c\in\prod_i(\D_i\cup\{m_i\})$ where
$c_i<m_i$ for at least one $i$, $M^*$ only buys the item with non-zero probability
from a player $i$ with $c_i<m_i$ (the same proof approach works).
Let $\tD:=\prod_{i}\bD_i$ and $\tD_{-i}:=\prod_{j\neq i}\bD_j$; also, let
$\tD_i:=\bD_i:=\D_i\cup\{m_i\}$ for uniformity of notation.  
For $c\in\tD$, define $\Om(c)=\{\w\in\Om: \w_i=\es\text{ for all $i$ s.t. $c_i=m_i$}\}$, 
if there is some $i$ such that $c_i<m_i$, and $\Om$ otherwise.
In our LP \eqref{primal}, or its relaxation \newlp (where we move to an LP-relaxation of
the \cm problem and consider the allocation-set $\Om_{\ext}$), we now enforce
\eqref{e2}--\eqref{e45} for all $i$, all $c_i,c'_i\in\tD_i$ and all $c_{-i}\in\tD_{-i}$.  

Let \indprim, and \indnewlp (with allocation-set $\Om_\lp$) denote these new LPs.
When $n$ is a constant, both LPs have a polynomial number of constraints. So again by
considering the dual, we can efficiently compute: (i) an optimal solution to \indprim
given an optimal algorithm for the \cm problem; and (ii) an optimal solution to
\indnewlp. If the \cm problem can be encoded via \eqref{cmlp} and we have an LP-relative
approximation algorithm for the \cm problem, then one can use the rounding procedure
described in Theorem~\ref{bicredn2} to round the optimal solution to \indnewlp to a 
near-optimal solution to \indprim; the arguments are essentially identical.

So suppose that we have a near-optimal solution $(x,p)$ to \indprim.
We 
extend $(x,p)$ to a (DSIC-in-expectation, IR) mechanism $M=\bigl(\A,\{q_i\}\bigr)$ without
increasing the expected disutility. 
Here, $\A(c)$ and $q_i(c)$ denote as before the
allocation-distribution and expected payment to $i$, on input $c$. 

Define $y_c=\sum_\w x_{c,\w}\w$, where we treat $\w$ as a vector in $\{0,1\}^n$ with
$\w_i\equiv\w_{i,v}$ for the single covering object $v\in\T_i$.
Let $0\leq c_i^1<c_i^2<\ldots<c_i^{k_i}=c_i^{\max}$ be the values in $\D_i$, and set 
$c_i^{k_i+1}:=m_i$.  Define the mapping $H:C\to\tD$ as follows: set
$H(c):=\bigl(H_i(c_i)\bigr)_{i=1,\ldots,n}$, where $H_i(c_i)$ is $c_i^{r+1}$ if
$c_i\in(c_i^{r},c_i^{r+1}],\ r\leq k_i$, and $m_i$ if $c_i\geq m_i$. 
Define $H_{-i}(c_{-i}):=\bigl(H_j(c_j)\bigr)_{j\neq i}$. 

Consider $c\in C$. If $c_i\leq c_i^{\max}$ for at least one $i$,
we set $\A(c)=y_{H(c)}$. 
If $c_i>c_i^{\max}$ for all $i$, we set $\A(c)$ as in the VCG mechanism.
Since we are in the single-dimensional setting, if we show that for all 
$i$, $c_{-i}\in C_{-i}$, $\A(c)_i$ is non-increasing in $c_i$ and hits 0 at some point,
then we know that setting $q_i(c)=c_i\A(c)_i+\int_{c_i}^\infty\A(t,c_{-i})_i dt$ 
ensures such that $M=\bigl(\A,\{q_i\}\bigr)$ is (DSIC, IR)-in-expectation.

Consider some $i$, $c_{-i}\in C_{-i}$. 
If $c_j\leq c_j^{\max}$ for some $j\neq i$, then $\A(c)=y_{H(c)}$.
Since $H_i$ is non-decreasing in $c_i$ and $y_{c,i}$ is non-increasing in $c_i$ (which is 
easily verified), it follows that $\A(c)_i$ is non-increasing in $c_i$.
Also, if $c_{-i}\in\tD_{-i}$, then one can argue as in the proof of Theorem~\ref{bicredn2}
that $q_i(c)\leq p_{i,c}$. Hence, $M$ has expected total payment at most
$\sum_{c,i}\Pr_\D(c)p_{i,c}$.  
Suppose $c_j>c_j^{\max}$ for all $j\neq i$. Then, $H_j(c_j)=m_j$ for all $j\neq i$. So
$\A(c)=y_{H(c)}$ for $c_i\leq c_i^{\max}$, and is the VCG allocation for
$c_i>c_i^{\max}$. Therefore, $\A(c)_i=1$ for $c_i\leq c_i^{\max}$, and the VCG allocation
for $c_i>c_i^{\max}$, which is clearly non-increasing in $c_i$.

\begin{theorem} \label{dsicthm}
For single-dimensional problems with a constant number of players, we obtain the
same guarantees as in Table~\ref{onedapps}, but under the stronger solution concept of
DSIC-in-expectation and IR. 
\end{theorem}

\section*{Acknowledgments}
We thank the anonymous reviewers and referees of the conference and journal submissions
for various useful comments. 
Section~\ref{weaker} was prompted by questions raised by an anonymous referee of the
journal submission regarding the solution concepts \eqref{opta}-\eqref{optc}.


\appendix

\section{Proof of Lemma~\ref{estim}} \label{append-prelim}
Consider the following LP, which is the same as \eqref{primal} except that we only
consider $c\in\bigcup_i(\D_i\times\D_{-i})$. 
\begin{alignat}{3}
\min & \quad & \sum_{c\in\D}{\textstyle \Pr_\D}(c)
\Bigl(\sum_iq_{i,c}&+\kp\sum_{\w}x_{c,\w}\pub(\w)\Bigr) \tag{LP} \label{eprim} \\
\text{s.t.} 
&& \sum_{\w}x_{c,\w} & = 1 \qquad && \forall c \in {\bigcup_i(\D_i\times\D_{-i})} \label{ee2} \\
&& q_{i,(c_i,c_{-i})} - \sum_{\w}c_i(\w)x_{(c_i,c_{-i}),\w} & \geq 
q_{i,(c'_i,c_{-i})} -\sum_{\w}c_i(\w)x_{(c'_i,c_{-i}),\w} \ \
&& \forall i,c_i,c'_i\in\D_i,c_{-i}\in\D_{-i} \label{ee3} \\
&& q_{i,(c_i,c_{-i})} - \sum_{\w}c_i(\w)x_{(c_i,c_{-i}),\w} & \geq 0 \qquad 
&& \forall i, c_i\in\bD_i, c_{-i}\in\D_{-i} \label{ee4} \\
&& q,x  & \geq 0. \label{ee45}
\end{alignat}

Let $M=\bigl(\A,\{p_i\}\bigr)$ be an optimal mechanism. Recall that $\iopt$ is the
expected disutility of $M$.
Then, $M$ naturally yields a feasible solution $(x,q)$ to \eqref{eprim} of
objective value $\iopt$, where $x_{c,\w}=\Pr_M[\A(c)=\w]$ and $q_{i,c}=\E_M[p_i(c)]$.
Let $(\hx,\hq)$ be an {\em optimal basic solution} to \eqref{eprim}. 
Then, for some $N$ such that $\log N$ is polynomially bounded in the input size $\I$,
we can say that the values of all variables are integer multiples of $\frac{1}{N}$, and
$\log (N\hx_{c,\w}), \log (N\hq_{i,c})=\poly(\I)$ for 
all $i$, $c\in\bigcup_i(\D_i\times\D_{-i})$, $\w$.

First, we claim that we may assume that for every
$i, c_i\in\D_i, c_{-i}\in\D_{-i}$, if whenever $\hx_{c,\w}>0$ we have $\w_i=\es$
(where $c=(c_i,c_{-i})$), then  $\hq_{i,c}=0$. If not, then \eqref{ee3} implies that
$\hq_{i,(\tc_i,c_{-i})}-\sum_\w\tc_i(\w)x_{(\tc_i,c_{-i}),\w}\geq\hq_{i,c}$ for all
$\tc_i\in\D_i$ and decreasing $\hq_{i,(\tc_i,c_{-i})}$ by $\hq_{i,c}$ for all
$\tc_i\in\D_i$ continues to satisfy \eqref{ee3}--\eqref{ee45}.

Set $m_i:=\max\bigl(2\sum_{i,v\in\T_i}\max_{c_i\in\D_i}c_{i,v},N\sum_{i,c}\hq_{i,c}\bigr)$
for all $i$. So $\log m_i=\poly(\I)$. 
Recall that $\bD_i:=\D_i\cup\{m_i\bon_{\T_i}\}$ for all $i\in[n]$, 
and $\bD:=\bigcup_i(\bD_i\times\D_{-i})$. 

Now we extend $(\hx,\hq)$ to $(\tx,\tq)$ that assigns values also to type-profiles  
in $\bD\sm\bigcup_i(\D_i\times\D_{-i})$ so that constraints
\eqref{ee2}--\eqref{ee45} hold for all $i$, $c_i,c'_i\in\bD_i$, $c_{-i}\in\D_{-i}$.
First set $\tx_{c,\w}=\hx_{c,\w},\ \tq_{i,c}=\hq_{i,c}$ for all $i$, $\w$, 
$c\in\bigcup_i(\D_i\times\D_{-i})$. 
Consider $c\in\bD\sm\bigcup_i(\D_i\times\D_{-i})$, and let $i$ be such that
$c_i=m_i\bon_{\T_i}$ (there is exactly one such $i$). 
We ``run'' VCG on $c$ considering only the cost incurred by the players. 
That is, we set $\tx_{c,\w}=1$ for 
$\w=\w(c):=\arg\min_{\w\in\Om}\sum_i c_i(\w)$ and pay 
$\tq_{i,c}=\min_{\w\in\Om: \w_i=\es}\sum_jc_j(\w)-\sum_{j\neq i}c_j\bigl(\w(c)\bigr)$ to
each player $i$. 
Note that the choice of $m_i$ ensures that $\w(c)_i=\es$ and hence, $\tq_{i,c}=0$. 

We claim that this extension satisfies \eqref{ee2}--\eqref{ee45} for all 
$i$, $c_i,c'_i\in\bD_i$, $c_{-i}\in\D_{-i}$.
Fix $i$, $c_i,c'_i\in\bD_i$, $c_{-i}\in\D_{-i}$. 
It is clear that \eqref{ee2}, \eqref{ee45} hold.
If $c_i\in\D_i$, then \eqref{ee4} clearly holds; if $c_i=m_i\bon_{\T_i}$, then it 
again holds since $\tx_{c,\w}=1$ for $\w=\w(c)$ and $\w(c)_i=\es$.
To verify \eqref{ee3}, we consider four cases. If $c_i,c'_i\in\D_i$, then \eqref{ee3}
holds since $(\tx,\tq)$ extends $(\hx,\hq)$. If $c_i=c'_i=m_i\bon_{\T_i}$, then
\eqref{ee3} trivially holds. If $c_i\in\D_i,\ c'_i=m_i\bon_{\T_i}$, then \eqref{ee3} holds
since the RHS of \eqref{ee3} is 0 
(as $\tx_{(c'_i,c_{-i}),\w(c'_i,c_{-i})}=1$ and $\w(c'_i,c_{-i})_i=\es$).
We are left with the case $c_i=m_i\bon_{\T_i}$ and $c'_i\in\D_i$. If
whenever $\tx_{(c'_i,c_{-i}),\w}=\hx_{(c'_i,c_{-i}),\w}>0$ we have $\w_i=\es$, then we
also have $\tq_{i,(c'_i,c_{-i})}=\hq_{i,(c'_i,c_{-i})}=0$ by our earlier claim, so the RHS of
\eqref{ee3} is 0, and \eqref{ee3} holds. Otherwise, we have 
$\sum_\w c_i(\w)\tx_{(c'_i,c_{-i}),\w}\geq\frac{m_i}{N}\geq\hq_{i,(c'_i,c_{-i})}$, so
the RHS of \eqref{ee3} is at most 0, and \eqref{ee3} holds.

Thus, we have shown that $(\tx,\tq)$ is a feasible solution to \eqref{primal}. Now we can
apply Theorem~\ref{multiround} to extend $(\tx,\tq)$ and obtain a
support-based-(IC-in-expectation, IR) mechanism $M^*$ whose expected disutility is at most   
$\sum_{c,i}\Pr_{\D}(c)\bigl(\tq_{i,c}+\sum_{\w}\tx_{c,\w}\pub(\w)\bigr)\leq\iopt$. 
Since $\tx_{(m_i\bon_{\T_i},c_{-i}),\w}>0$ implies that $\w_i=\es$ for all $i$, $M^*$
satisfies the required conditions. \hfill\qedsymbol

\section{Inferiority of \boldmath $k$-lookahead procurement auctions} 
\label{append-lookahead} 
The following $k$-lookahead auction was proposed by~\cite{DobzinskiFK11} for the
single-item revenue-maximization problem generalizing the 1-lookahead auction considered 
by~\cite{Ronen,RonenS}: on input $v=(v_1,\ldots,v_n)$, pick the set $I$ of $k$ players
with highest values, and run the revenue-maximizing (DSIC, IR) mechanism for player-set
$I$ where the distribution we use for $I$ is the conditional distribution of the values
for $I$ given the values $(v_i)_{i\notin I}$ for the other players. Dobzinski et
al.~\cite{DobzinskiFK11} show that the $k$-lookahead auction achieves a constant-fraction
of the revenue of the optimal (DSIC, IR) mechanism.

For any $k\geq 2$, we can consider an analogous definition of $k$-lookahead auction for
the single-item procurement problem: on input $c$, we pick the set $I$ of $k$ players with
smallest costs, and run the payment-minimizing support-based-(IC-in-expectation, IR) mechanism
for $I$ for the conditional distribution of $I$'s costs given $(c_i)_{i\notin I}$. We call
this the {\em $k$-lookahead procurement auction}. 
The following example shows that the expected total payment of the $k$-lookahead
procurement auction can be arbitrarily large, even when $k=n-1$ (that is, we drop only 1 
player), and compared to the optimal expected total payment of even a deterministic (DSIC, 
IR) mechanism.  

Let $t=K+\e$ where $\e>0$, and $\dt>0$.
The distribution $\D$ consists of $n$ points: each $c$ in 
$\{c: c_n=t,\ \ \exists i\in[n-1]\text{ s.t. }c_i=0,\ \ c_j=K\ \forall j\neq i,n\}$ 
has probability $\Pr_{\D}(c)=\frac{1-\dt}{n-1}$, and 
the type-profile $c$ where $c_i=K\ \forall i\neq n,\ c_n=t$ has probability
$\Pr_{\D}(c)=\dt$.  

Let $k=n-1$. The $k$-lookahead procurement auction will always select the player-set 
$I=\{1,\ldots,n-1\}$, and the conditional distribution of values of players in $I$ is
simply $\D$. Let $M'$ be the support-based-(IC-in-expectation, IR) mechanism for the players in
$I$ under this conditional distribution $\D$.
Suppose that on input $(K,K,\ldots,K,t)$, the $k$-lookahead auction (which runs $M'$) buys
the item from player $i\in I$ with probability $x_i$. Clearly, $\sum_{i\in I}x_i=1$. Then,
on the input $\tc$ where $\tc_i=0$, $\tc_j=K$ for all $j\neq i,n$, $\tc_n=t$, the
mechanism must also buy the item from player $i$ with probability at least $x_i$ since
$M'$ is support-based IC. So since $M'$ is support-based (IC-in-expectation, IR), the payment to
player $i$ under input $\tc$ is at least $K$, and therefore the expected total payment of
the $k$-lookahead auction is at least 
$K(\sum_{i\in I}x_i)\cdot\frac{1-\dt}{n-1}=\frac{K(1-\dt)}{n-1}$. 

Now consider the following mechanism $M=\bigl(\A,\{p_i\}\bigr)$. Consider input $c$. If
some player $i<n$ has $c_i=0$, $M$ buys the item from such a player $i$ (breaking ties in
some fixed way). Otherwise, if $c_n\leq t$, $M$ buys from player $n$; else, $M$ buys from
the player $i$ with smallest $c_i$. 
It is easy to verify that for every $i$ and $c_{-i}$, this allocation rule $\A$ is
monotonically decreasing in $c_i$. 
Let $p_i(c)=0$ if $M$ does not buy the item from $i$ on input $c$, and 
$\max\{z: \text{$M$ buys the item from $i$ on input $(z,c_{-i})$}\}$ otherwise.    
By a well-known fact, (see, e.g., Theorem 9.39 in~\cite{agt}), $M$
is DSIC and IR. Then, the total payment under any $c\in\D$ for which $c_i=0$ for some $i$, 
is 0, and the total payment under the input where $c_i=K$ for all $i\neq n$, $c_n=t$ is
$t$. So the expected total payment of $M$ is $t\dt$.

Thus the ratio of the expected total payments of $M'$ and $M$ is at least
$\frac{K(1-\dt)}{t\dt (n-1)}$, which can be made arbitrarily large by choosing $\dt$ and
$\e=t-K$ small enough.

\section{Insights for revenue-maximization in packing domains} \label{append-packing}  
Our study of \paym problems also leads to some interesting insights into the
revenue-maximization problem in packing settings with correlated players. 
We state two results 
that are obtained via relatively-simple observations, 
but we believe are nevertheless of interest.

In Section~\ref{append-packextn}, we justify our comment in the Introduction that the
problem of extending an LP solution to a suitable mechanism becomes much easier in a
packing setting such as combinatorial auctions (CAs). 
We show that any solution to an LP-relaxation similar to \eqref{primal} for the
revenue-maximization problem in CAs can be extended to a (DSIC-in-expectation, IR)
mechanism without any loss in revenue. 
In Section~\ref{append-packapx}, we obtain a noteworthy extension of a result
in~\cite{DobzinskiFK11}. We show 
that in single-dimensional packing settings, a $\rho$-approximation algorithm for the SWM
problem can be used to obtain a $\rho$-approximation (DSIC-in-expectation, IR)-mechanism
for the revenue-maximization problem. 
We obtain this by noting that a $\rho$-approximation for the SWM problem can be 
used to obtain a $\rho$-approximate separation oracle for the dual of the
revenue-maximization LP, despite the fact that the separation problem is an SWM problem
possibly involving {\em negative-valued} inputs%
\footnote{Recall that negative costs are quite problematic for the \cm problem, and hence,
we had to resort to the method outlined in Section~\ref{single}.},
which then yields, in a fairly-standard way, a $\rho$-approximate solution to the
revenue-maximization LP.

\subsection{Extending LP solutions to (DSIC-in-expectation, IR)
  mechanisms}   \label{append-packextn} 
We consider the prototypical problem of combinatorial auctions; similar arguments
can be made for other packing domains. In CAs, a feasible 
allocation $\w$ is one that allots a disjoint set $\w_i$ of items (which could be
empty) to each player $i$, and player $i$'s value under allocation $\w$ is $v_i(\w_i)$,
where $v_i:2^{[m]}\mapsto\R_+$ is player $i$'s private valuation function. We use
$v_i(\w)$ to denote $v_i(\w_i)$. 
Let $V_i$ denote the set of all private types of player $i$, and $V_{-i}=\prod_{j\neq i}V_j$.
As before, let $\Om$ be the set of all feasible solutions.

We consider the following LP along the lines of \eqref{primal}.
Since we are in a packing setting, we do not need the $m_i$
estimates. We may assume that each $\D_i$ contains the valuation $0_i$, where $0_i(\w)=0$
for all $\w$, since if not, we can just add this to $\D_i$, and set $\Pr_\D(0_i,v_{-i})=0$.
Let $\bD':=\bigcup_i(\D_i\times\D_{-i})$. 
\begin{alignat}{3}
\max & \quad & \sum_{v\in\D}{\textstyle \Pr_\D(v)}&\sum_ip_{i,v} \tag{R-P} \label{packp} \\
\text{s.t.} && \sum_{\w\in\Om}x_{v,\w} & \leq 1 \qquad && \forall v\in {\bD'} \label{pe2} \\
&& \sum_{\w\in\Om}v_i(\w)x_{(v_i,v_{-i}),\w} - p_{i,(v_i,v_{-i})} & \geq 
\sum_{\w\in\Om}v_i(\w)x_{(v'_i,v_{-i}),\w} - p_{i,(v'_i,v_{-i})} \quad  
&& \forall i, v_i,v'_i\in\D_i, v_{-i}\in\D_{-i} \notag \\ 
&& \sum_{\w\in\Om}v_i(\w)x_{(v_i,v_{-i}),\w} - p_{i,(v_i,v_{-i})} & \geq 0 \qquad 
&& \forall i, v\in\bD' \notag \\ 
&& p,x  & \geq 0. \notag 
\end{alignat}

The above LP-relaxation is similar to the LP in~\cite{DobzinskiFK11} for
single-dimensional packing problems. For CAs, the LP in~\cite{DobzinskiFK11}
is subtly different from \eqref{packp}: their allocation variables 
encode the probability that a player $i$ receives a set $S$ of items. 
If we let the allocation space in \eqref{packp} be the set $\Om_\ext$ of extreme points of
the standard 
LP-relaxation for the CA problem, then our formulations coincide since a feasible solution
to the standard LP specifies the extent to which each player receives each set. 
(The convex-decomposition technique in~\cite{LaviS11} directly implies that an
integrality-gap verifying $\rho$-approximation algorithm for the SWM problem can be 
used to decompose the fractional allocation specified by the LP in~\cite{DobzinskiFK11}
scaled by $\rho$ into a distribution over $\Om$, and thereby obtain a solution to
\eqref{packp}.)  

Let $(\tx,\tp)$ be a solution to \eqref{packp}.
We convert $(\tx,\tp)$ to a (DSIC-in-expectation, IR) mechanism
$M=\bigl(\A,\{q_i\}\bigr)$ with no smaller expected total revenue.
Here, $\A(v)$ and $q_i(v)$ are the allocation-distribution and expected price of player
$i$ on input $v$. 
Since any support-based-(IC, IR)-in-expectation mechanism yields a feasible solution to
\eqref{packp}, this also shows that 
{\em any support-based-(IC, IR)-in-expectation mechanism for CAs can be extended to a
(DSIC-in-expectation, IR) mechanism without any loss in revenue}.  
 
Our argument is similar to that in the proof of Theorem~\ref{multiround}, but the packing
nature of the problem simplifies things significantly. 
We may assume that 
all constraints \eqref{pe2} are tight; otherwise if there is some $v\in\bD'$ with
$\sum_{\w\in\Om}\tx_{v,\w}<1$, then letting $\w^0$ be allocation where $\w^0_i=\es$ for
all $i$, we can increase $\tx_{v,\w_0}$ to make this sum 1 without affecting feasibility.

First, we set $\A(v)=\tx_v,\ q_i(v)=\tp_{i,v}$
for all $v\in\bD'$ and all $i$, so it is clear that the expected total revenue of $M$ is
the value of $(\tx,\tp)$. 

If $|\{i: v_i\notin\D_i\}|\geq 2$, then we give
everyone the empty-set and charge everyone 0. Otherwise, suppose 
$v_i\notin\D_i,\ v_{-i}\in\D_{-i}$. 
Let $\bv^{(i)}=\arg\max_{\tv_i\in\D_i}\bigl(\sum_\w v_i(\w)\tx_{(\tv_i,v_{-i}),\w}-\tp_{i,(\tv_i,v_{-i})}\bigr)$ 
and $\by^{(i)}=\tx_{(\bv^{(i)},v_{-i})}$.
For $\w\in\Om$, let $\proj_i(\w)$ denote the allocation where player $i$ receives
$\w_i\sse[m]$, and the other players receive $\es$.
Viewing $\A(v)$ as the random variable specifying the allocation selected, 
we set $\A(v)=\proj_i(\w)$ with probability $\by^{(i)}_\w$. We set $q_i(v)=\tp_{i,(\bv^{(i)},v_{-i})}$.
Since $0_i\in\D_i$, we have $\E_\A[v_i(\A(v))]-q_i(v)\geq
\sum_\w 0_i(\w)\tx_{(0_i,v_{-i}),\w}-\tp_{i,(0_i,v_{-i})}\geq 0$, so $M$ is
IR-in-expectation.

To see that $M$ is DSIC in expectation, consider some $i$, $v_i,v'_i\in V_i$, $v_{-i}\in V_{-i}$. 
If $v_{-i}\notin\D_{-i}$, then player $i$ always receives the empty set and pays 0.
Otherwise, we have ensured by definition that player $i$ does not benefit by lying.

\subsection{Utilizing approximation algorithms for the SWM problem}  
\label{append-packapx}
We briefly sketch how to utilize a $\rho$-approximation algorithm for the SWM problem to
obtain a $\rho$-approximation (DSIC-in-expectation, IR)-mechanism for the
revenue-maximization problem in single-dimensional settings. 
Similar arguments apply to packing settings with additive
types. 

Consider again the LP \eqref{packp}. For all $i$, $\w$, we now have 
$v_i(\w)=v_i\al_{i,\w}$, where $v_i\in\R$ is $i$'s private type and $\al_{i,\w}\geq 0$ is
public knowledge. This is identical to the LP in~\cite{DobzinskiFK11} for single-parameter
revenue-maximization problems. 
It will be convenient to view $\w\in\Om$ as the vector
$\{\al_{i,\w}\}_{i\in[n]}\in\R_+^n$. Since we are in a packing setting, $\Om$ is downward
closed, so $\w\in\Om$ and $\w'\leq\w$ implies that $\w'\in\Om$.
The dual of \eqref{packp} is:
\begin{alignat}{3}
\min & \quad && \qquad \sum_v \ga_v \tag{R-D} \label{packd} \\
\text{s.t.} & && 
\sum_{i:v\in\D_i\times\D_{-i}}\Bigl(\sum_{v'_i\in\D_i}\bigl(v_i(\w)y_{i,(v_i,v_{-i}),v'_i} && -
v'_i(\w)y_{i,(v'_i,v_{-i}),v_i}\bigr)+v_i(\w)\be_{i,v}\Bigr) \notag \\[-1.25ex]
& && && \leq \ga_v \qquad \qquad \qquad \qquad \ \ \forall v\in{\bD'}, \w\in\Om \label{pe5} \\
& && \sum_{v'_i\in\D_i}\bigl(y_{i,(v_i,v_{-i}),v'_i} - y_{i,(v'_i,v_{-i}),v_i}\bigr) && +\ \be_{i,v_i,v_{-i}}
\geq{\textstyle \Pr_D}(v) 
\qquad \forall i, v\in\bD' \label{pe6} \\ 
& && y,\be,\gm \geq 0. \label{pe7}
\end{alignat}

\newcommand{\optpack}{\ensuremath{\mathsf{optr}}}

Let $\optpack$ be the common optimal value of \eqref{packp} and \eqref{packd}.
Define $\tht_i^v=\sum_{v'_i\in\D_i}(v_iy_{i,(v_i,v_{-i}),v'_i}-v'_iy_{i,(v'_i,v_{-i}),v_i})+v_i\be_{i,v}$
if $v\in\D_i\times\D_{-i}$, and $0$ otherwise.
Then, the separation problem for \eqref{packd} amounts to determining if 
$\max_{\w\in\Om}\sum_i\tht_i^v\al_{i,\w}\leq\gm_v$ for every $v\in\bD'$; that is, 
solving an SWM problem over the allocation space $\Om$ under the input
$\tht^v=\{\tht^v_i\}_{i\in[n]}$ for every $v$. 
Given a $\rho$-approximation algorithm $\A$ for the SWM
problem (where $\rho\geq 1$) that only works with nonnegative inputs, we can obtain a
$\rho$-approximate 
solution for the input $\tht^v$ as follows. Set $(\tht^+)^v_i:=\max(0,\tht^v_i)$, use $\A$
with the input $(\tht^+)^v=\{(\tht^+)^v_i\}_{i\in[n]}$ to obtain a solution $\w'$, and let
$\w''_i=\w'_i$ if $\tht_i\geq 0$ and $0$ otherwise. 
Thus, $\A$ can be used to approximately separate over constraints \eqref{pe5}.
We now argue that this implies that we can obtain a solution to \eqref{packp} of value at
least $\optpack/\rho$. This follows a routine argument in the approximation-algorithms
literature (see, e.g.,~\cite{JainMS03}).

Define $\Pc(\nu):=\{(y,\beta,\gm): \eqref{pe5}\text{--}\eqref{pe7},\ \sum_v\gm_v\leq\nu\}$. 
Note that $\optpack$ is the smallest $\nu$ such that $\Pc(\nu)\neq\es$. 
Given $\nu, y,\beta,\gm$, we either show that
$(y,\beta,\gm\rho)\in\Pc(\nu\rho)$, or we exhibit a hyperplane separating $(y,\beta,\gm)$ from
$\Pc(\nu)$. To do this, we first check if $\sum_v\gm_v\leq\nu$,
\eqref{pe6}, \eqref{pe7} hold, and if not use the appropriate inequality as the
separating hyperplane. Next, for every $v\in\bD'$, we use $\A$ as specified above to obtain some
$\w''\in\Om$. If in this process, the LHS of \eqref{pe5} exceeds $\gm_v$ for some $v$,
then we return the corresponding inequality as the separating hyperplane.
Otherwise, for all $v\in\bD'$ and all $\w\in\Om$, the LHS of \eqref{pe5} is at most
$\gm\rho$, and so $(y,\beta,\gm\rho)\in\Pc(\nu\rho)$.

Thus, for a fixed $\nu$, in polynomial time, the ellipsoid
method either certifies that $\Pc(\nu)=\es$, or returns a point $(y,\beta,\gm)$ with
$(y,\beta,\gm\rho)\in\Pc(\nu\rho)$. 
We find the smallest value $\nu^*$ (via binary search) such that the ellipsoid method run
for $\nu^*$ (with the above separation oracle) returns a solution $(y^*,\beta^*,\gm^*)$ with
$(y^*,\beta^*,\gm^*\rho)\in\Pc(\nu^*\rho)$; hence, $\optpack\leq\nu^*\rho$. For any
$\e>0$, running the ellipsoid method for $\nu^*-\e$ yields a polynomial-size certificate
for the emptiness of $\Pc(\nu^*-\e)$. This consists of the polynomially many violated
inequalities returned by the separation oracle during the execution of the ellipsoid
method and the inequality $\sum_v\gm_v\leq\nu^*-\e$. By duality (or
Farkas' lemma), this means that here is a polynomial-size solution $(\tx,\tp)$ to
\eqref{packp} whose value is at least $\nu^*-\e$. Taking $\e$ to be 
$1/\exp(\text{input size})$ (so $\ln\bigl(\frac{1}{\e}\bigr)$ is polynomially bounded),
this also implies that $(\tx,\tp)$ has value at least $\nu^*\geq\optpack/\rho$. 

We can now convert $(\tx,\tp)$ to a (DSIC-in-expectation, IR) mechanism with revenue at
least $\optpack/\rho$ via the procedure described in Section~\ref{append-packextn}.

\end{document}